\documentclass[letterpaper]{sig-alternate-10pt}
\pdfoutput=1

\usepackage{rotating}
\usepackage{amsfonts}
\usepackage{amssymb}
\usepackage{amsmath}
\usepackage[official]{eurosym}
\newtheorem{theorem}{Theorem}[section]

\newtheorem{corollary}{Corollary}[section]
\newtheorem{proposition}{Proposition}[section]
\newtheorem{lemma}{Lemma}[section]
\newtheorem{definition}{Definition}[section]

\newcommand{\ie}[0]{\textit{i.e.}, }
\newcommand{\eg}[0]{\textit{e.g.}, }
\newcommand{\etal}[0]{\textit{et al.}}

\def\sharedaffiliation{%
\end{tabular}
\begin{tabular}{c}}

\begin{document}

\title{Leveraging information in vehicular parking games}

\numberofauthors{3} %  in this sample file, there are a *total*

\author{
\alignauthor Evangelia Kokolaki\\
       \email{evako@di.uoa.gr}
\alignauthor Merkouris Karaliopoulos\\
       \email{mkaralio@di.uoa.gr}
\alignauthor Ioannis Stavrakakis\\
       \email{ioannis@di.uoa.gr}
%} \and  % use '\and' if you need 'another row' of author names
%
\sharedaffiliation
 \affaddr{Department of Informatics and Telecommunications}\\
 \affaddr{National and Kapodistrian University of Athens, Athens, Greece}
 %\affaddr{Athens, Greece}
}

%       \affaddr{National \& Kapodistrian University of Athens\\Dept. Informatics \& Telecommunications}\\
%       \affaddr{Ilissia, 157 84}\\
%       \affaddr{Athens, Greece}\\

% Just remember to make sure that the TOTAL number of authors
% is the number that will appear on the first page PLUS the
% number that will appear in the \additionalauthors section.

%\author{Paper ID number: 292}

\maketitle
%\begin{abstract}
%Our paper approaches the parking assistance service in urban environments as an instance of service provision in non-cooperative network environments. We propose normative abstractions for the way drivers pursue parking space and the way they respond to \emph{partial} or \emph{complete information} for parking demand and supply as well as specific pricing policies on public and private parking facilities. The drivers are viewed as strategic agents who make rational decisions attempting to minimize the cost of the acquired parking spot. We formulate the resulting games as instances of \emph{resource selection games} and derive their equilibria under various expressions of uncertainty for the overall parking demand. The efficiency of the equilibrium states is compared against the optimum assignment that could be determined by a centralized entity and conditions are derived for minimizing the related \emph{price of anarchy} value. Our results provide theoretical insights to the dynamics emerging from the drivers' behavior as well as useful hints for the pricing and practical management of on-street and private parking resources. On top of that, they contribute further theoretical understanding of the effective information mechanisms within competitive contexts.
%\end{abstract}

\begin{abstract}
Our paper approaches the \emph{parking assistance service} in urban environments as an instance of service provision in non-cooperative network environments. We propose normative abstractions for the way drivers pursue parking space and the way they respond to \emph{partial} or \emph{complete information} for parking demand and supply as well as specific pricing policies on public and private parking facilities. The drivers are viewed as strategic agents who make rational decisions attempting to minimize the cost of the acquired parking spot. We formulate the resulting games as \emph{resource selection games} and derive their equilibria under different expressions of uncertainty about the overall parking demand. The efficiency of the equilibrium states is compared against the optimal assignment that could be determined by a centralized entity and conditions are derived for minimizing the related \emph{price of anarchy} value. Our results provide useful hints for the pricing and practical management of on-street and private parking resources. More importantly, they exemplify counterintuitive \emph{less-is-more effects} about the way information availability modulates the service cost, which underpin general competitive service provision settings and contribute to the better understanding of effective information mechanisms.
\end{abstract}

%% A category with the (minimum) three required fields
%%A category including the fourth, optional field follows...
%\category{D.2.8}{Software Engineering}{Metrics}[complexity measures, performance measures]


\keywords{Parking assistance service, vehicular networks, parking games, uncertainty, price of anarchy} % NOT required for Proceedings

\section{Introduction}\label{Introduction}

In various mobile applications, networked entities (\ie agents) are called to autonomously decide on how to best coexist with each other in the network, \ie \emph{cooperate with} and/or \emph{compete against} eachother, to optimally serve their interests.
%ensure that their networking role will serve primarily their interest.
The agents' co-action may actually take various forms and pertain to different network functions depending on the particular network paradigm. For example, in autonomic networks, each agent (node) is called to decide whether to dispose or not its own scarce resources (\ie energy, bandwidth and storage space) in favor of others' welfare, anticipating their support in due course. Other instances explicitly discriminate between the resource/service provider and resource/service consumer; namely, there is a network-external operator that manages the service provision and a number of user nodes that seek to get access to it at minimum cost. This paper attempts to delineate and explore the dynamics that arise in this last type of competitive settings.

A crucial determinant for these dynamics is the information different nodes possess about the service resource availability and the demand for it.
%Under this notion of competition for external service, user nodes may take advantage of any information source to diminish the %competition artefacts.
Indeed, any such information becomes an asset that shapes the nodes' behavior and modulates their incentive to compete. The  information factor affects the final outcome of nodes' interactions, and, eventually, the benefit that is accrued by them as well as the service provider, \eg her income when she charges her service. Technically, this information may be announced centrally, even by the service provider herself, or opportunistically collected by and distributed among the network of service consumers.

These competitive contexts are well captured in auction-based frameworks \cite{adwords}\cite{Akyildiz06}. In general, sellers-auctioneers draw on auction mechanisms to allocate both divisible and non-divisible resources among multiple agents, with the aim to maximize either their own revenue or the social welfare.
%Potential buyers express their interest in the auctioned resources via their bids, anticipating to be awarded at a minimum cost. The auctioneer collects the requests and bids of the agents and determines who is endowed with resource and at what cost. Advanced and novel auction-based mechanisms in competitive networking applications, such as spectrum allocation services \cite{Akyildiz06} and online sponsored search engines \cite{adwords}, have started to gain ground as promising solutions for resolving the competition.
Typically, the auctioneer avails private information on the auction set-up that, when published, can modulate the bidders' strategies, escalating or moderating competition, and hence determining the outcome of the auction procedure, \ie resource winners and their payments \cite{Milgrom82}\cite{Simon09}.
In this paper we study another instance of competitive service provision involving vehicular nodes within urban environments: the \emph{parking assistance service}. Vehicle drivers seek and compete for the cheaper but scarce on-street parking space, while
the parking service provider aims at maximizing the parking capacity utilization and his revenue.
%\textbf{[de dinetai emfasi sto kerdos, giati oi doyleies poy paratithentai den yponooyn kapoia pedio agoras/pwlisis]},
%whereas drivers wish for cost-effective parking space within their, usually overlapped, area of interest, after short search time.
As with auctions, the information about the resources and the demand for them may vary and shape the behavior of competing agents. On the one hand, the parking service provider may collect and broadcast different amounts of information to the drivers; whereas, vehicles may exploit wireless communication and information sensing technologies to gain themselves partial knowledge about the location and/or vacancy of parking spots.

%or otherwise, opportunistically distributed among vehicles, giving rise to misbehaving phenomena.
The way the opportunistic exchange of information among vehicles may sharpen competition is studied in \cite{Kokolaki11} and \cite{Delot09}. In \cite{Kokolaki11}, Kokolaki \etal~simulate a fully cooperative opportunistic parking space assistance scheme, whereby each parking spot is equipped with a sensor device providing information about its occupancy status.
%Vehicles, properly equipped with short-range wireless interfaces and adequate storage and processing capacity, collect and exchange information on the location and status of each spot they encounter.
It is shown that the full exchange of information upon encounters of vehicles may give rise to synchronization effects (vehicles are steered towards similar locations), sharpen competition, and eventually render the search process inefficient. Anticipating this effect, Delot \etal~propose in \cite{Delot09} a distributed virtual parking space reservation mechanism, whereby vehicles vacating a parking spot selectively distribute this information to their proximity. Hence, they mitigate the competition for the scarce parking spots by controlling the diffusion of the parking information among drivers.
%In particular, according to their light reservation protocol, every time a vehicle vacates a parking spot communicates this event to drivers within his range. Among the potential answers he may receives from the drivers, he first chooses a single one, applying distance criteria, and then acknowledges his decision to the particular driver. Although the ``reservation'' cannot be considered reliable due to the autonomic nature of the networking nodes and the distributed operation of the protocol, it shows to reduce the competition, compared to pure $V2V$ information dissemination networks.

%Motivated by the concerns that the processes of information dissemination (benefiting service discovery) and competition (reducing the service delivery prospects) are coupled and counter-acting, we elaborate the parking search application in a broader view, formalizing the competition effects within non-cooperative networking environments. The questions we pose are:
%
%\begin{itemize}
%  \item how different types of information (complete or partial) on the parking demand and supply modulate drivers' incentive to compete,
%  \item how the information impacts on the cost that incurs to the drivers and the revenue accruing to the parking operator,
%  \item which key-parameters optimize drivers' decisions for their activity,
%  \item what principles run the effective parking information mechanisms.
%\end{itemize}
Drawing on the parking search assistance service, our paper seeks to systematically explore a broader phenomenon, evidenced in several instances of service provision within non-cooperative networking environments: the double-edged impact of information on the overall service efficiency, \ie its assistance with resource/service discovery, on the one hand, and the sharpening of competition for it, on the other. Questions we address are: How do different types of information (complete or partial) on the parking demand and supply modulate drivers' incentive to compete? How does such information affect the cost that drivers incur and the revenue accruing for the parking service operator?

%The significant parking demand in city centers make city councils to draw on both public and private parking facilities to respond to the parking needs of the car volumes that daily visit popular in-city destinations. Under the conventional parking search practice, drivers are called to decide either to compete for the cheaper but scarce on-street parking spots or head for the more expensive private parking lot. In the first case, they run the risk of failing to get a spot and having to \emph{a posteriori} take the more expensive alternative, this time suffering the additional cost in time, fuel consumption (and stress) of the failured attempt. In the absence of, possibly centrally-driven, coordination, the drivers pursue selfishly to minimize the cost of access to parking facilities. However, the intuitive decision to head for the cheaper or free-of-cost on-street parking space, combined with the scarcity of public parking capacity in urban curbside of typical center areas, give rise to \emph{tragedy of commons} effects \cite{Hardin68} and highlight the game-theoretic dynamics behind the parking spot selection problem.

We take a game-theoretic approach and view the drivers as rational selfish agents that pursue to minimize the cost they will pay for acquired parking space. The drivers choose to either compete for the cheaper but scarce on-street parking spots or head for the more expensive private parking lot. In the first case, they run the risk of failing to get a spot and having to \emph{a posteriori} take the more expensive alternative, this time suffering the additional \emph{cruising} cost in terms of time, fuel consumption (and stress) of the failured attempt. Drivers make their decisions drawing on information of variable accuracy about the parking demand (number of drivers) and supply (number of parking spots and pricing policy), which is broadcast from the parking service operator. With this common knowledge at hand, drivers react rationally seeking to minimize the cost of their decisions. The announced information impacts on the resulting driver interaction and ultimately the total cost paid. Thus, its systematic manipulation provides useful hints for the realization of effective centralized information mechanisms.

We formulate the parking spot selection problem as an instance of resource selection games, abstracting from spatial and temporal variations in parking demand and supply, in Section \ref{sec:game}. We then analyze the game variant with complete information about parking demand in Section \ref{sec:cinfo}, where we derive the equilibrium behaviors of the drivers and compare the induced social cost against the optimal one via the Price of Anarchy metric, leaving proofs for the Appendix. We relax the assumption for the availability of complete information and derive the corresponding analysis in Section \ref{sec:incinfo}. Indeed, in Section \ref{sec:results}, we show that the optimization of the equilibrium social cost is feasible by properly choosing the charging cost and the location of the private parking facilities. Less intuitively, assessing the impact of information, we present less-is-more phenomena arguing that partial information maximizes drivers benefit, compared to complete knowledge. We outline related research in Section \ref{Related_work} and we close the discussion in Section \ref{sec:discussion}, drawing parallels between the game-theoretic assertions for drivers' behavior and insights from the cognitive psychology domain.

\section{The parking spot selection game}\label{sec:game}

In the parking spot selection game, the set of players consists of drivers who circulate within the center area of a big city in search of parking space. Typically, in these regions, parking is completely forbidden or constrained in whole areas of road blocks so that the real effective curbside is significantly limited (see Fig. \ref{fig:park_map}). The drivers have to decide whether to drive towards the scarce low-cost (controlled) public parking spots or the more expensive private parking lot (we see all local lots collectively as one). All parking spots that lie in the same public or private area are assumed to be of the same value for the players --we discuss this assumption further in Section \ref{sec:discussion}. Thus, the decisions are made on the two \emph{sets} of parking spots rather than individual set items. The two sets jointly suffice to serve all parking requests.

%%mobicom
%\begin{figure*}[t]
%%\begin{center}
%\vspace{-15pt}
%\hspace{9pt}%130
%\begin{tabular}{cc}
%  \includegraphics[scale=0.3,angle=90]{parking_plan}
%  &
%  \includegraphics[scale=0.25,angle=90]{google_map}
%  \end{tabular}
%  \\
%%\begin{scriptsize} Parking area  \end{scriptsize}
%%\end{center}
%\vspace{-10pt}
%\caption{Parking and geographical map (thanks to Google maps) of the centre area of a big European city. Dashed lines show metered controlled (public) parking spots, whereas ``P'' denotes private parking facilities. The map illustrates, as well, the capacity of both parking options.\label{fig:park_map}}
%\end{figure*}

\begin{figure*}[t]
%\begin{center}
\vspace{-15pt}
\hspace{9pt}%130
\begin{tabular}{cc}
  \includegraphics[scale=0.3,angle=90]{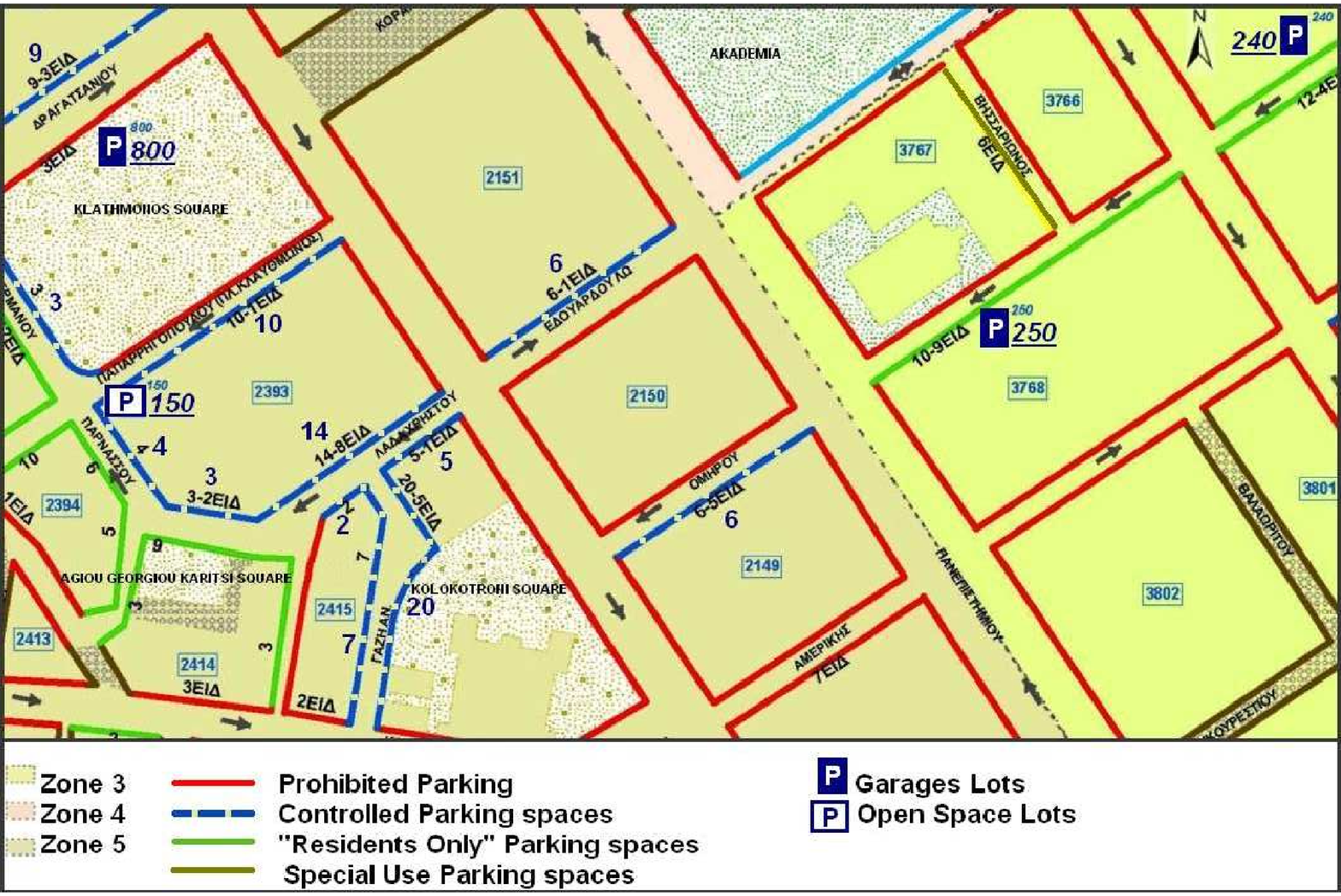}
  &
  \includegraphics[scale=0.25,angle=90]{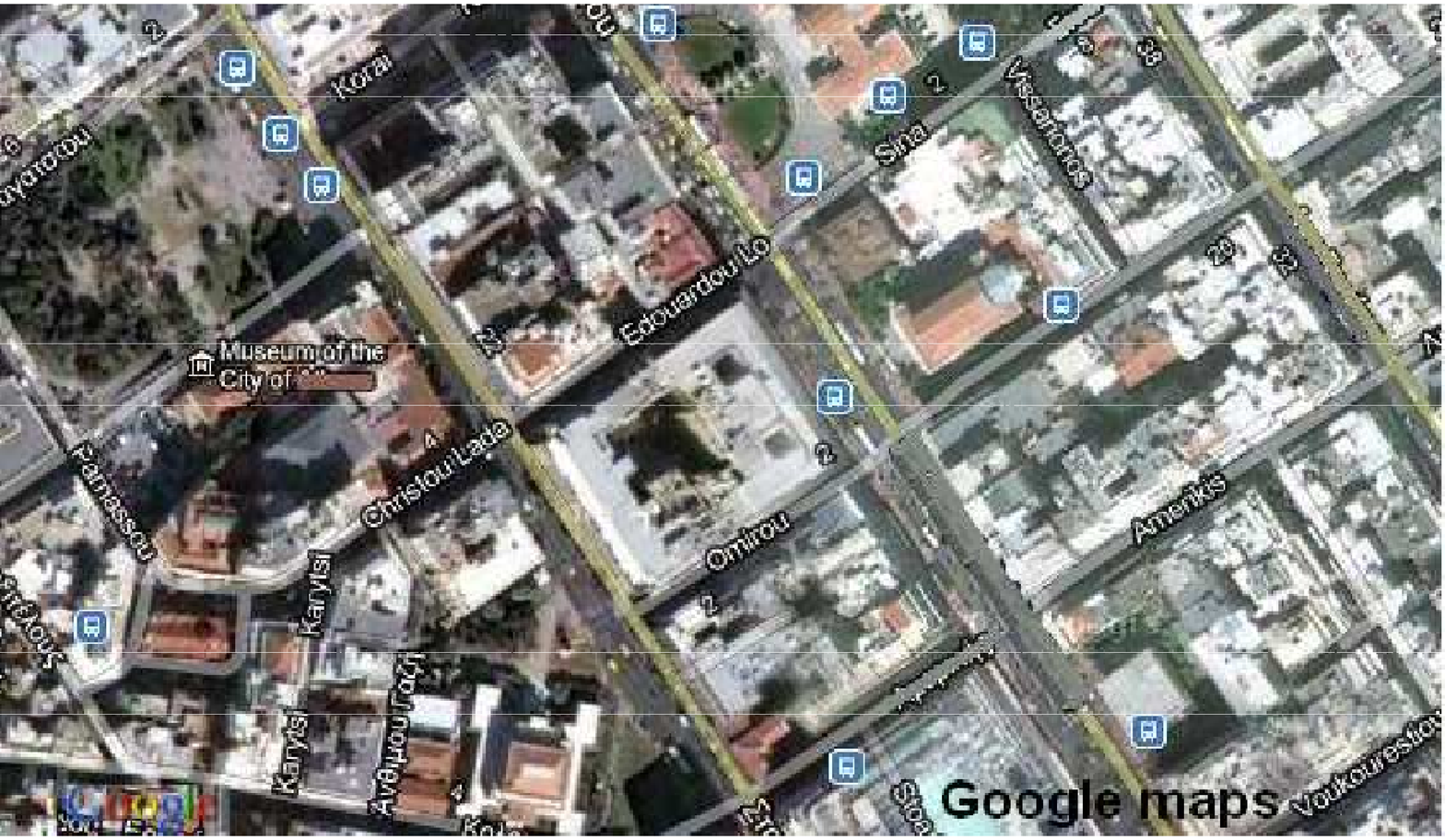}
  \end{tabular}
  \\
%\begin{scriptsize} Parking area  \end{scriptsize}
%\end{center}
\vspace{-12pt}
\caption{Parking and geographical map (thanks to Google maps) of the centre area of a big European city. Dashed lines show metered controlled (public) parking spots, whereas ``P'' denotes private parking facilities. The map illustrates, as well, the capacity of both parking options.\label{fig:park_map}}
\vspace{-2pt}
\end{figure*}

We observe drivers' behavior within a particular time window over which they reach this parking area. In general, these synchronization phenomena in drivers' flow occur at specific time zones during the day \cite{cityeu}. Herein, we account for morning hours or driving in the area for business purposes coupled with long parking duration. Thus, the collective decision making on parking space selection can be formulated as an instance of the strategic \emph{resource selection games}, whereby $N$ players (\ie drivers) compete against each other for a finite number of common resources (\ie public parking spots) \cite{Ashlagi06}. More formally, the one-shot parking spot selection game is defined as follows:

%Particularly, the symmetric resource selection game is defined as follows:
%\newdef{definition}{Definition}
%\begin{definition}\label{def:resource_game}
%A symmetric \emph{Resource Selection Game} is a tuple $\Gamma(N)=(R,(w_{j})_{j=1}^{R})$, where:
% \begin{itemize}
%\item $\mathcal{N}=\{1,...,N\}, N>1$ is the set of players,
%\item $\mathcal{R}=\{1,...,R\}$, $R\geq1$, the set of resources,
%\item $A_{i}=\mathcal{R}$, the set of possible actions,
%\item $w_{j}:\mathcal{N}\rightarrow\mathbb{R}$, the cost function for every user of resource j,
% \item $c_{i}^{N}(x)=w_{a_{i}}(\sigma_{a_{i}}(x))$, the cost function for every player i, where $x=(a_{i},a_{-i})\in\times_{i=1}^{N}A_{i}$ is the action profile and $\sigma_{a_{i}}(x)$ the number of players with action $a_{i}\in{A_{i}}$.
% \end{itemize}
%\end{definition}

%\newdef{definition}{Definition}
\begin{definition}\label{def:parking_game}
A \emph{Parking Spot Selection Game} is a tuple
$\Gamma(N)=(\mathcal{N},\mathcal{R},(w_{j})_{j\in(pub,priv)})$, where:
\begin{itemize}
\item $\mathcal{N}=\{1,...,N\}$, $N>1$ is the set of drivers who seek for parking space,
\item $\mathcal{R}=\mathcal{R}_{pub}\cup\mathcal{R}_{priv}$ is the set of parking spots; $\mathcal{R}_{pub}$ is the set of public spots, with $R=|\mathcal{R}_{pub}|\geq1$; $\mathcal{R}_{priv}$ the set of private spots, with $|\mathcal{R}_{priv}|\geq N$,
\item $A_{i}=\{public, private\}$, is the action set for each driver $i\in \mathcal{N}$,
\item $w_{pub}()$ and $w_{priv}()$ are the cost functions of the two actions, respectively\footnote{Note that the cost functions are defined over the action set of each user; in the original definition of resource selection games in \cite{Ashlagi06}, cost functions are defined over the resources but the resource set coincides with the action set.}.
\end{itemize}
\end{definition}

The parking spot selection game comes under the broader family of \emph{congestion games}\footnote{Readers who are more familiar with game theory will notice a resemblance to the atomic variant of Pigou's selfish routing example \cite{Pigou20}. Pigou's paths correspond to the two parking alternatives, one having a high user-independent use cost and the other a cost that scales with the number of users (albeit not linearly).}. %The players are identical with respect to the game formulation.
The players' payoffs (here: costs) are non-decreasing functions of the \emph{number} of players competing for the parking capacity rather than their identities and common to all players.
%hence, they are determined by \emph{what} is played rather than \emph{who} plays it and are common to all  players.
%means that they share the same action space while the cost to play a specific action is shaped from the action profile (namely, the actions being played) and not from the identity of participants in this profile.
More specifically, drivers who decide to compete for the public parking space undergo the risk of not being among the $R$ winner-drivers to get a public spot. In this case, they have to eventually resort to private parking space, only after wasting extra time and fuel (plus patience supply) on the failed attempt. The expected cost of the action $public$, $w_{pub}:A_{1}\times ...\times A_{N}\rightarrow\mathbb{R}$, is therefore a function of the number of drivers $k$ taking it, and is given by

\vspace{-5pt}
\small
\begin{equation}
w_{pub}(k) = min(1,R/k)c_{pub,s}+(1-min(1,R/k))c_{pub,f}
\end{equation}
\normalsize
where $c_{pub,s}$ is the cost of successfully competing for public parking space, whereas $c_{pub,f} = \gamma\cdot c_{pub,s}, \gamma>1$, is the cost of competing, failing, and eventually paying for private parking space.

On the other hand, the cost of private parking space is fixed

\vspace{-7pt}
\small
\begin{equation}
w_{priv}(k) = c_{priv} = \beta\cdot c_{pub,s}
\end{equation}
\normalsize

where $1 < \beta < \gamma$, so that the excess cost $\delta\cdot c_{pub,s}$, with $\delta=\gamma-\beta>0$,
%As far as the pricing scheme is concerned, $cost_{pub}$ and $cost_{priv}=a \cdot cost_{pub}$, $a>1$ are the fixed costs for %public and private parking space, respectively. If a player chooses to drive towards the limited public space, he risks to pay %an amount that exceeds both costs. In particular, everyone who attempts unsuccessfully to reserve public parking space, is %enforced to use private parking space at an augmented cost, that is, $cost_{fail}=b \cdot cost_{pub}$, $b>a$.
%Particularly, the excess price $(b-a)\cdot cost_{pub}$
reflects the actual cost of cruising and the ``virtual'' cost of wasted time till eventually heading to the private parking space.
Figure \ref{fig:cost_fun} plots the cost functions against the number of drivers, for both parking options, under different charging schemes.
%%%%mobicom
%\begin{figure}[btp]
%\vspace{-15pt}
%\begin{center}
%\begin{tabular}{cc}
%\hspace{-15pt}
%  \includegraphics[scale=0.3]{costs}
%  \hspace{-10pt}
%  \includegraphics[scale=0.3]{costs2}
%\end{tabular}
%\end{center}
%\vspace{-15pt} \caption{The cost functions for public and private parking space: $R=50$, $c_{pub,s}=1$.\label{fig:cost_fun}}
%\vspace{-10pt}
%\end{figure}

\begin{figure}[btp]
\vspace{-100pt}
\begin{center}
\begin{tabular}{cc}
\hspace{-40pt}
  \includegraphics[scale=0.3]{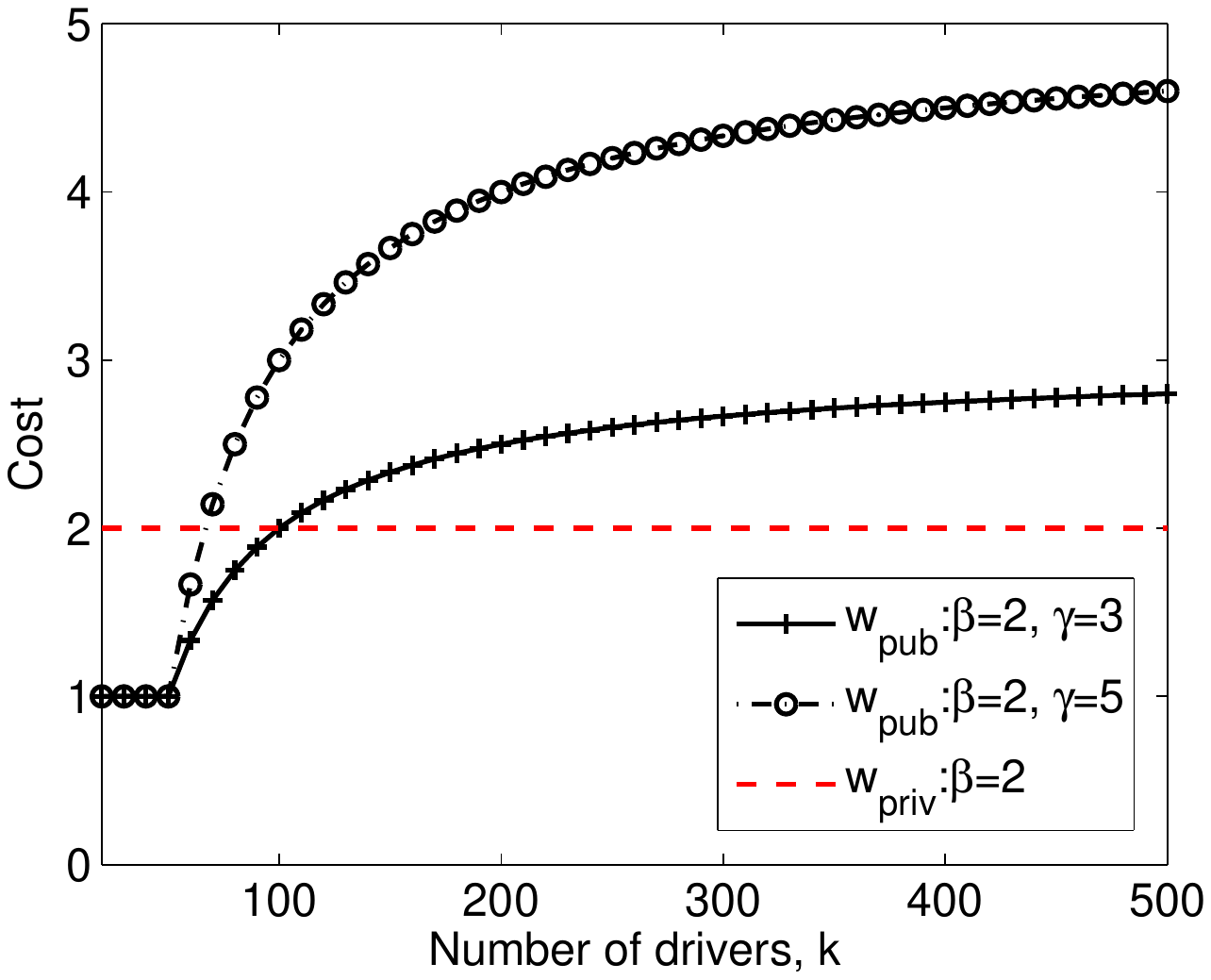}
  \hspace{-60pt}
  \includegraphics[scale=0.3]{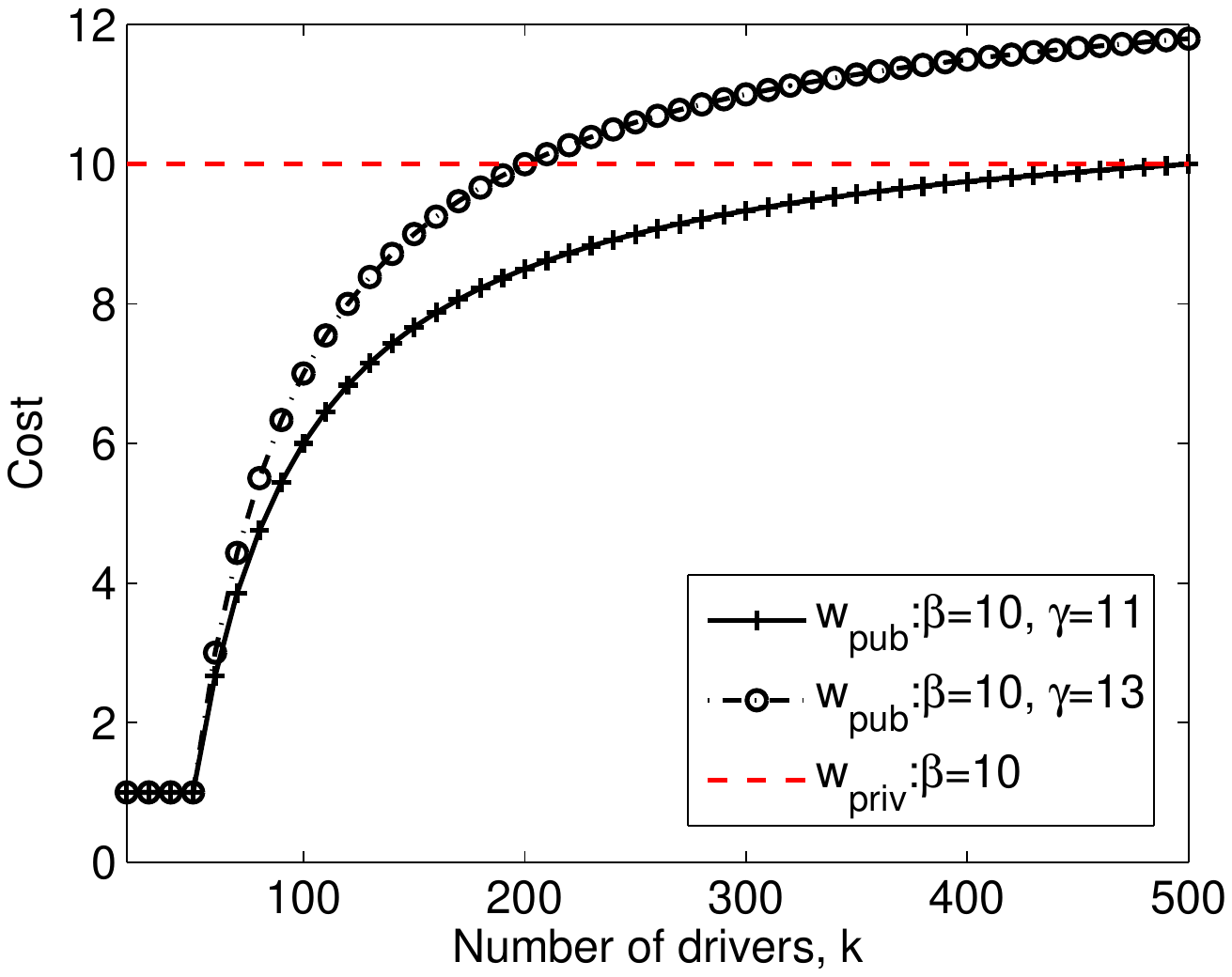}
\end{tabular}
\end{center}
\vspace{-85pt} \caption{The cost functions for public and private parking space: $R=50$, $c_{pub,s}=1$.\label{fig:cost_fun}}
\vspace{-10pt}
\end{figure}

We denote every action profile with the vector $a=(a_{i},a_{-i})\in\times_{k=1}^{N}A_{k}$, where $a_{-i}$ denotes the actions of all other drivers but player $i$ in the profile $a$. Besides the two \emph{pure} strategies coinciding with the pursuit of public and private parking space, the drivers may also randomize over them.
%More precisely, the expected cost for every player $i\in\mathcal{N}$ who chooses to risk for public parking space depends on the number of risk-takers, that is, $c_{i}^{N}(x)=w_{pub}(\sigma_{pub}(x),cost_{pub},cost_{fail})$, where $x=(public,a_{-i})$ and $\sigma_{pub}(x)$ the total number of drivers who choose the public parking space in this profile. On the contrary, those who renege for private parking space pay the fixed cost $c_{i}^{N}(x)=w_{priv}$, with $x=(private,a_{-i})$ (\ie the cost of private parking $cost_{priv}$), irrespective of their number. Therefore, all players who share the same preference for the parking service are called to pay the same amount.
%The players may also perform according to a specific probability distribution applied in their action set.
In particular, if $\Delta(A_{i})$ is the set of probability distributions over the action set of player $i$, a player's \emph{mixed action} corresponds to a vector $p=(p_{pub},p_{priv})\in \Delta(A_{i})$, where $p_{pub}$ and $p_{priv}$ are the probabilities of the pure actions, with $p_{pub}+p_{priv}=1$, while its cost is a weighted sum of the %are the probabilities with which the driver decides in favor of public and private parking space, respectively. Every mixed action is paired with an expected cost over
cost functions $w_{pub}()$ and $w_{priv}()$ of the pure actions.

In this game-theoretic formulation, the drivers are assumed to be rational strategic players. They explicitly consider the presence of identical counter-actors that also make rational decisions, weight the costs related to every possible action profile, and act as cost-minimizers. In doing so, they usually do not avail precise information about the actual demand, \ie competition, for parking resources.
%This decisional operation should take into account the additional cost that they will have to pay if they fail to reserve an
%available public spot.
%Particularly, the players draw on every information reflects the overall parking demand, provided that they are aware of the parking supply. However, the innate dynamic nature of the game does not allow for precise information and hence fine measurements. Only probabilistic information may feed their estimations.
In the following sections, we analyze the parking selection game under different levels of uncertainty for the overall parking demand, ranging from the highly optimistic scenario of complete knowledge to one of high uncertainty about it. In all cases, we look into both the stable and optimal operational conditions and the respective costs incurred by the players.
\section{Complete knowledge of parking demand}\label{sec:cinfo}

Ideally, the players determine their strategy under complete knowledge of those parameters that shape their cost. Given the symmetry of the game, the additional piece of information that is considered available to the players, besides the number of vacant parking spots and the employed pricing policy, is the level of parking demand, \ie the number of drivers searching for parking space. We draw on concepts from \cite{Koutsoupias09} and theoretical results from \cite{Ashlagi06,Cheng04} to derive the equilibrium strategies for the game $\Gamma(N)$ and assess their (in)efficiency.
%In our case, the only critical factor is the number of other players, $N-1$. This assumption formulates the strategic variant %of the game, where each player infers his opponents' cost, given the total action profile. Thus, without any uncertainty %constraints, the investigation of this abstraction of the parking spot selection game takes advantage of classic normative %paradigms. Following, we apply known theoretical results to derive the equilibrium strategies for the competing players, in %light of the Definition \ref{def:parking_game}.

\subsection{Pure Equilibria strategies}\label{sec:pe}

%\hspace{12pt}
\textbf{Existence:} The parking spot selection game constitutes a symmetric game, where the action set is common to all players and consists of two possible actions, $public$ and $private$. Cheng \etal~ have shown in (\cite{Cheng04}, Theorem 1) that every symmetric game with two strategies has an equilibrium in pure strategies.
%\newdef{theorem}{Theorem}
%\begin{theorem}\label{thm:symmetric_two_strategies}
%A symmetric game with two strategies has an equilibrium in pure strategies.
%\end{theorem}
%
%\begin{definition}\label{def:tragedy_of_commons}
%The \emph{Tragedy of Commons} refers to a dilemmatic situation where a number of rational independent and self-interested agents deplete a shared limited resource. The problem formulation is stated as follows:
%\begin{itemize}
%  \item $\mathcal{N}=\{1,...,N\}, N>1$ is the set of players,
%  \item $A_{i}\in\{use,not~use\}$, the set of possible actions,
%  \item $w:\mathcal{N}\rightarrow\mathbb{R}$, the non-increasing cost function for every user of the resource,
%  \item $cost$, the fixed cost for everyone who refuses to use the resource.
%\end{itemize}
%\end{definition}
%

\textbf{Computation:} Thanks to the game's symmetry, the full set of $2^N$ different action profiles maps into
$N+1$ different action meta-profiles. Each meta-profile $a(m), m \in [0,N]$ encompasses all ${N\choose m}$ different action profiles that result in the same number of drivers competing for on-street parking space. The expected costs for these $m$ drivers and for the $N-m$ ones choosing directly the private parking lot alternative are functions of $a(m)$ rather than the exact action profile.

In general, the cost $c_i^N(a_i,a_{-i})$ for the driver $i$ under the action profile $a=(a_i,a_{-i})$ is

%\small
%\vspace{-5pt}
%%\begin{figure}[h!]
%\begin{eqnarray}
%%\[
%c_{i}^{N}(a_{i},a_{-i})=\left\{
%\begin{array}{l l}
%w_{pub}(\sigma_{pub}(a_{i},a_{-i}))=
%\;P_{pub} \cdot cost_{pub} + (1-P_{pub}) \cdot cost_{fail},\\
%\; \; \; \; \quad\text{for~ $a_{i}=public$}\\
%w_{priv}=cost_{priv}, \quad\text{for~ $a_{i}=private$}\\
%\end{array} \right.
%%\]
%\label{equ:cost}
%\vspace{-5pt}
%\end{eqnarray}
%
\vspace{-10pt}
\small
%\begin{figure}[h!]
\begin{eqnarray}
%\[
c_{i}^{N}(a_{i},a_{-i})=\left\{
\begin{array}{l l}
w_{pub}(\sigma_{pub}(a)),~~~~~~~~~~\text{for~ $a_{i}=public$} \\
w_{priv}(N-\sigma_{pub}(a)),~~~ \text{for~ $a_{i}=private$} \\
\end{array} \right.
%\]
\label{equ:cost}
\label{eqn:costN}
\end{eqnarray}
\normalsize

where $\sigma_{pub}(a)$ is the number of competing drivers for on-street parking under action profile $a$.
%The core of the decision mechanism in the strategic games lies in the cost function that is paired with every action option. Following the Definition \ref{def:parking_game}, in the parking spot selection game the choice to renege from competition is assessed at a fixed cost (\ie $w_{priv}=cost_{priv}$) while the risk-takers pay according to the cost function $w_{pub}$, which equals the expected cost for public space reservation, given the number of competitors $\sigma_{pub}(a_{i},a_{-i})=\sum_{k=1,a_{k}=public}^{N} 1$ or the probability to park in public space $P_{pub}=\min(\frac{R}{\sigma_{pub}(a_{i},a_{-i})},1)$.
%\end{figure}
%\normalsize
Equilibria action profiles combine the players' \emph{best-responses} to their opponents' actions. Formally, the action profile $a=(a_{i},a_{-i})$ is a pure Nash equilibrium if for all $i\in \mathcal{N}$:

\vspace{-10pt}
%\small
\begin{equation}\label{equ:pe}
a_{i}\in \arg \min_{a'_{i}\in A_{i}}(c_{i}^{N}(a'_{i},a_{-i}))
\end{equation}
\normalsize
%
%The Definition \ref{equ:pe} describes formally which outcome of the interaction between several strategic decision makers constitutes a Nash equilibrium.
%\begin{definition}\label{equ:pe}
%The strategy profile $x=(a_{i},a_{-i})$ is a pure Nash equilibrium if for all $i\in \mathcal{N}$
%\begin{equation}
%a_{i}\in \arg \min_{a'_{i}\in A_{i}}(c_{i}^{N}(a'_{i},a_{-i}))
%\end{equation}
%\end{definition}
%Thus, if the interaction between several strategic decision makers constitutes a Nash equilibrium,
so that no player has anything to gain by changing her decision unilaterally.

Therefore, to derive the equilibria states, we locate the conditions on $\sigma_{pub}$ that break the equilibrium definition and reverse them. More specifically, given an action profile $a$ with $\sigma_{pub}(a)$ competing drivers, a player gains by changing her decision to play action $a_{i}$ in two circumstances:

\vspace{-10pt}
%\small
\begin{equation}\label{eq:cond1}
\hspace{6.7pt}
 ~when ~a_{i}=private \text{ and } w_{pub}(\sigma_{pub}(a)+1)<c_{priv}
\end{equation}
\vspace{-20pt}
\begin{equation}\label{eq:cond2}
\hspace{-30pt}
~when ~a_{i}=public \text{ and } w_{pub}(\sigma_{pub}(a))> c_{priv}
\end{equation}
\normalsize

%\begin{eqnarray}
%\text{when } a_{i}=private \text{ and } w_{pub}(\sigma_{pub}(a)+1)<c_{priv} \nonumber \label{eq:cond1} \\
%\text{when } a_{i}=public \text{ and } w_{pub}(\sigma_{pub}(a))> c_{priv} \label{eq:cond2}
%\end{eqnarray}
%
%In the first case, a driver profits from changing her decision to pay for private parking since the expected cost exceeds the expected cost when competing for %on-curb parking. The second case occurs when a risk-taker reneges for the safer option.
Taking into account the relation between the number of drivers and the available on-street parking spots, $R$, we can postulate the following Lemma:

%\newdef{lemma}{Lemma}
\begin{lemma}
In the parking spot selection game $\Gamma(N)$, a driver is motivated to change his action $a_{i}$ in the following circumstances:
%\small
\begin{equation}\label{eq_pe:1}
\hspace{-50pt}
\bullet ~a_{i}=private \text{ and } \sigma_{pub}(a)<R\leq N~\text{or }
\end{equation}
\vspace{-15pt}
\begin{equation}\label{eq_pe:2}
\hspace{68pt}
R\leq\sigma_{pub}(a)<\sigma_{0}-1\leq N~\text{or }
\end{equation}
\begin{equation}\label{eq_pe:3}
%\hspace{-51pt}
\hspace{18pt}
\sigma_{pub}(a)<N\leq R
%\sigma_{pub}<N \text{, if } N\leq R \text{,}
\end{equation}
%
%\begin{equation}\label{eq_pe:4}
%\hspace{5.1pt}
%\bullet ~a_{i}=public \text{ and } \sigma_{pub}>\sigma_{0} \text{, if } R<\sigma_{pub}\leq N \text{,}
%\end{equation}
\begin{equation}\label{eq_pe:4}
\hspace{-46pt}
\bullet ~a_{i}=public \text{ and } R<\sigma_{0}<\sigma_{pub}(a)\leq N
\end{equation}

where $\sigma_{0}=\frac{R(\gamma-1)}{\delta}\in\mathbb{R}$.
\label{le:pe}
\end{lemma}
\normalsize
\begin{proof}
Conditions (\ref{eq_pe:1}) and (\ref{eq_pe:3}) are trivial. Since the current number of competing vehicles is less than the on-street parking capacity, every driver having originally chosen the private parking option
%in the fear of the extra penalty cost (\ie in case of belated arrival at the public parking spots),
has the incentive to change her decision due to the price differential between $c_{pub,s}$ and $c_{priv}$.

When $\sigma_{pub}(a)$ exceeds the public parking supply, as in (\ref{eq_pe:2}), a driver who has decided to avoid competition, profits from switching her action when the expected cost of playing $public$ becomes less than the fixed cost of playing $private$. From (\ref{eqn:costN}) and (\ref{eq:cond1}), it must hold that:
%\begin{equation}
%w_{pub}(\sigma_{pub}+1)<w_{priv} \nonumber
%\end{equation}

\vspace{-8pt}
\small
\begin{eqnarray}
\frac{R}{\sigma_{pub}(a)+1}\cdot c_{pub,s}+(1-\frac{R}{\sigma_{pub}(a)+1})\cdot c_{pub,f} < c_{priv} \Rightarrow  \nonumber \\
\sigma_{pub}(a) < \frac{R(\gamma-1)}{\delta}-1 \nonumber
\end{eqnarray}
\normalsize
%that is,
%\begin{equation}
%\sigma_{pub}<\frac{R(\gamma-1)}{\delta}-1 \nonumber
%\end{equation}

which yields (\ref{eq_pe:2}).

On the contrary, a driver that first decides to compete for public parking space, switches to $private$ if the competing drivers outnumber the public parking resources. Namely, from (\ref{equ:cost}) and (\ref{eq:cond2}), when

\vspace{-8pt}
\small
\begin{eqnarray}
%w_{pub}(\sigma_{pub}(a))>w_{priv}\longrightarrow \nonumber \\
\frac{R}{\sigma_{pub}(a)}\cdot c_{pub,s}+(1-\frac{R}{\sigma_{pub}(a)})\cdot c_{pub,f}> c_{priv} \Rightarrow \nonumber \\
\sigma_{pub}(a) > \frac{R(\gamma-1)}{\delta} \nonumber
\end{eqnarray}
\normalsize

inline with (\ref{eq_pe:4}).
\end{proof}

It is now possible to state the following Theorem for the pure Nash equilibria of the parking spot selection game. %Hereafter, we assume that $\sigma_{0}=\frac{R(\gamma-1)}{\delta}>R$.
%The Lemma \ref{le:pe} provides insights to the conditions - expressed via the number of risk-takers - that violate the %equilibrium principle. Therefore, the opposite conditions constitute Nash equilibria.

%\newdef{theorem}{Theorem}
\begin{theorem}\label{thm:pe}
A parking spot selection game has:
\begin{itemize}
\vspace{-3pt}
  \item one Nash equilibrium $a^*$ with $\sigma_{pub}(a^*)=\sigma_{pub}^{NE,1}=N$, if $N\leq\sigma_{0}$ and $\sigma_{0}\in\mathbb{R}$
  \item ${N\choose \lfloor\sigma_{0}\rfloor}$ Nash equilibrium profiles $a'$ with $\sigma_{pub}(a')=\sigma_{pub}^{NE,2}=\lfloor\sigma_{0}\rfloor$, if $N>\sigma_{0}$ and $\sigma_{0}\in (R,N)\backslash\mathbb{N^*}$
  \item ${N\choose \sigma_{0}}$ Nash equilibrium profiles $a'$ with $\sigma_{pub}(a')=\sigma_{pub}^{NE,2}=\sigma_{0}$ and ${N\choose \sigma_{0}-1}$ Nash equilibrium profiles $a^\star$ with $\sigma_{pub}(a^\star)=\sigma_{pub}^{NE,3}=\sigma_{0}-1$, if $N>\sigma_{0}$ and $\sigma_{0}\in [R+1,N]\cap\mathbb{N^*}$.
     % \vspace{-5pt}
\end{itemize}
%where $\sigma_{0}=\frac{R(b-1)}{b-a}$.
\end{theorem}

\begin{proof}
Theorem \ref{thm:pe} follows directly from (\ref{equ:pe}) and Lemma \ref{le:pe}. The game has two equilibrium conditions on $\sigma_{pub}$ for $N>\sigma_{0}$ with integer $\sigma_{0}$, or a unique equilibrium condition, otherwise.
%Theorem \ref{thm:pe} follows directly from (\ref{equ:pe}) and Lemma \ref{le:pe}.
%In particular, reversing conditions of (\ref{eq_pe:1})-(\ref{eq_pe:4}), we get the conditions corresponding to the equilibria states for $\sigma_{pub}(a)$. Furthermore, every permutation $\pi$ of an equilibrium profile $a$ with $\sigma_{pub}(\pi(a))=\sigma_{pub}(a)$, \ie member of the action meta-profile $a(\lfloor\sigma_{0}\rfloor)$, constitutes a Nash equilibrium as well, due to the game's symmetry.
\end{proof}

In the Appendix, we provide an alternative way to derive the equilibria of $\Gamma(N)$ via potential functions.

\textbf{Efficiency:} The efficiency of the equilibria action profiles resulting from the strategically selfish decisions of the drivers is assessed through the broadly used metric of the Price of Anarchy \cite{Koutsoupias09}.
%and Price of Stability \cite{Roughgarden02}.
It expresses the ratio of the social cost in the worst-case equilibria over the optimal social cost under ideal coordination of the driver's strategies.
%The decentralized model proposes a scalable mechanism for parking space assignments that counts against the efforts for %drivers' coordination. In fact, these efforts aspire to lower the competition level for the limited parking capacity and hence %the ultimate total cost paid by the parking space consumers. The direct comparison between the decentralized and the - %centralized - coordination mechanism is expressed in the in-famous \emph{Price of Anarchy} ratio of the highest cost paid in %equilibrium state over the minimum total cost. Hereafter, we devise the term \emph{social cost} to refer to the total amount %paid.

\begin{proposition}\label{prop:poa}
In the parking spot selection game, the \emph{pure Price of Anarchy} equals:
%\small
\begin{eqnarray}
\verb"PoA"=\left\{
\begin{array}{l l}
  \frac{\gamma N-(\gamma-1)\min(N,R)}{\min(N,R)+\beta\max(0,N-R)}, ~~~~if~ \sigma_{0} \geq N \nonumber \\ \\ \frac{\lfloor\sigma_{0}\rfloor\delta-R(\gamma-1)+\beta N}{R+\beta(N-R)}, ~~~~~~~~~~if~ \sigma_{0}<N \nonumber \\
  \end{array} \right.
\end{eqnarray}
\normalsize
%\hspace{5pt}
%where $\sigma_{0}=\frac{R(b-1)}{b-a}\in\mathbb{R}$.
\end{proposition}

\begin{proof}
The social cost under action profile $a$ equals:
%\small
%\begin{eqnarray}
%%\[
%C(\sigma_{pub}(a))=\sum_{i=1}^{N} c_{i}^{N}(a)=\left\{
%\begin{array}{l l}
%c_{pub,s}(N\beta-\sigma_{pub}(a)(\beta-1)), ~\text{if } \sigma_{pub}(a)\leq R \\
%c_{pub,s}(\sigma_{pub}(a)\delta-R(\gamma-1)+\beta N), \text{if } R <\sigma_{pub}(a)\leq N
%\end{array} \right.
%\label{eq:social_cost}
%%\]
%\end{eqnarray}
%\normalsize
%
\vspace{-8pt}
\small
\begin{eqnarray}
C(\sigma_{pub}(a))=\sum_{i=1}^{N} c_{i}^{N}(a)= \nonumber \\
c_{pub,s}(N\beta-\sigma_{pub}(a)(\beta-1)), ~\text{if } \sigma_{pub}(a)\leq R \text{ and }\\
c_{pub,s}(\sigma_{pub}(a)\delta-R(\gamma-1)+\beta N), ~\text{if } R <\sigma_{pub}(a)\leq N \nonumber
\label{eq:social_cost}
\end{eqnarray}
\normalsize

The numerators of the two ratios are obtained directly by replacing the first two $\sigma_{pub}^{NE}$ values (worst-cases) computed in Theorem \ref{thm:pe}. %and (\ref{eq:social_cost}), the social cost of the unique pure equilibrium profile of each value interval of $N$ equals the %corresponding numerators of \verb"PoA" ratios, multiplied by $cost_{pub}$. It should be noticed that all the pure-action %meta-profiles of $\sigma_{pub,eq}$ for $N>\sigma_{0}$ are identical in the social cost value.
On the other hand, under the ideal action profile $a_{opt}$, exactly $R$ drivers pursue on-street parking, whereas the remaining $N-R$ are served by the private parking resources. Therefore, under $a_{opt}$, no drivers find themselves in the unfortunate situation to have to pay the additional cost of cruising in terms of time and fuel after having unsuccessfully competed for an on-street parking spot. The optimal social cost, $C_{opt}$ is given by:
%is concerned, we need to find $\sigma_{pub,opt}

\vspace{-10pt}
\small
%\begin{equation}
%%\sum_{i=1}^{N} c_{i}^{N}(x_{opt})= \sigma_{pub,opt}\cdot w_{pub}(\sigma_{pub,opt})+ (N-\sigma_{pub,opt})\cdot w_{priv} \nonumber
%\end{equation}
%\normalsize
\begin{equation}
C_{opt} = \sum_{i=1}^{N} c_{i}^{N}(a_{opt})= c_{pub,s}[min(N,R)+\beta\cdot max(0,N-R)]
\end{equation}
%where $\sigma_{pub,opt}=\sigma_{pub}(x_{opt})$. In particular, the social optimum reaches its minimum value, which equals
%\small
%\begin{equation}
%cost_{pub}(\min(N,R)+a\max(0,(N-R))) \text{,~at~} \sigma_{pub,opt}=\min(N,R) \nonumber
%\end{equation}
\normalsize
\end{proof}

\begin{corollary}\label{cor:int_poa}
In the parking spot selection game, the \emph{pure Price of Anarchy} equals $\frac{1}{1-\frac{(\beta-1)R}{\beta N}}$, if $N>\sigma_{0}$ and $\sigma_{0}\in [R+1,N]\cap\mathbb{N^*}$.
\end{corollary}

\begin{proof}
From Theorem \ref{thm:pe}, for integer $\sigma_{0}$ and $N>\sigma_{0}$ there are two sets of equilibria profiles with $\sigma_{pub}^{NE,2}=\sigma_{0}$ and $\sigma_{pub}^{NE,3}=\sigma_{0}-1$. The social costs at these profiles are $c_{pub,s}\cdot N\beta$ and $c_{pub,s}\cdot(N\beta-\delta)$, respectively. Since $\beta > 1$ and $\delta > 0$, the highest social cost, which determines the \verb"PoA" ratio, is paid in the first case.
\end{proof}

\begin{proposition}\label{prop:bound_poa}
In the parking spot selection game, the \emph{pure Price of Anarchy} is upper-bounded by $\frac{1}{1-R/N}$ with $N>R$.
%$\frac{b}{a-1}$
\end{proposition}

\begin{proof}
From Proposition \ref{prop:poa}, when $N\leq\sigma_{0}$,

\vspace{-10pt}
\small
\begin{equation}
\verb"PoA"=\frac{\gamma N-R(\gamma-1)}{R+\beta(N-R)}
=\frac{\gamma-\delta\sigma_0/N}{\frac{R+\beta(N-R)}{N}}
\leq \frac{\beta}{\beta-\frac{R}{N}(\beta-1)}
%\frac{an}{an-r(a-1)}=
< \frac{1}{1-R/N}
%<\frac{1}{1-\frac{a-b}{1-b}}=\frac{b-1}{a-1}<\frac{b}{a-1}$.\\
%
\nonumber
\end{equation}
\normalsize
%If $N\leq R$, $\verb"PoA"=1<\frac{b}{a-1}$.\\
%
Similarly, when $N>\sigma_{0}$,

\vspace{-10pt}
\small
\begin{equation}
\verb"PoA"=\frac{\lfloor\sigma_{0}\rfloor\delta-R(\gamma-1)+\beta N}{R+\beta(N-R)}\leq
%\frac{\sigma_{0}(b-a)-R(b-1)+aN}{R+a(N-R)}=
%\frac{aN}{aN-R(a-1)}=
\frac{1}{1-R\frac{(\beta-1)}{N\beta}}
%<\frac{1}{1-\frac{R}{N}}\leq\frac{1}{1-\frac{a-b}{1-b}}=\frac{b-1}{a-1}$\\
<\frac{1}{1-R/N}
%<\frac{b}{a-1}$.
\nonumber
\end{equation}
\normalsize
\end{proof}

\subsection{Mixed-action equilibria strategies}\label{me}
We mainly draw our attention on \emph{symmetric} mixed-action equilibria since these can be more helpful in dictating practical strategies in real systems. Asymmetric mixed-action equilibria are discussed in the end of the Section.

\textbf{Existence:} Ashlagi, Monderer, and Tennenholtz proved in (\cite{Ashlagi06}, Theorem 1) that a unique symmetric mixed equilibrium exists for the broader family of resource selection games with more than two players and increasing cost functions. It is trivial to repeat their proof and confirm this result for our parking spot selection game $\Gamma(N)$, with $N>R$ and cost functions $w_{pub}(\cdot)$ and $w_{priv}(\cdot)$ that are non-decreasing functions of the number of players (increasing and constant, respectively).
%The parking spot selection game pertains to the broader category of finite symmetric games and thus follows theoretical results %originally introduced for general symmetric settings []. Moreover, the resource selection framework brings additional results %with respect to the equilibrium existence that qualify the generic conclusions, assisting the investigation of the particular %game [].

%\begin{theorem}\label{thm:ex_mixed}
%A finite symmetric game has a symmetric mixed-action equilibrium.
%\end{theorem}
\textbf{Computation:}
If we denote by
%\small
\begin{equation}
B(\sigma_{pub};N,p_{pub})={N\choose \sigma_{pub}}p_{pub}^{\sigma_{pub}}(1-p_{pub})^{N-\sigma_{pub}}
\end{equation}
\normalsize
the probability distribution of the number of drivers that decide to compete for on-street parking spots, where $p=(p_{pub},p_{priv})$ denotes a mixed-action, then

\vspace{-8pt}
\small
\begin{eqnarray}\label{eq:costs_sym}
    c_{i}^{N}(public,p) &=& \sum_{\sigma_{pub}=0}^{N-1}w_{pub}(\sigma_{pub}+1)B(\sigma_{pub};N-1,p_{pub}) \nonumber \\  c_{i}^{N}(private,p) &=& c_{priv} \nonumber
\end{eqnarray}
\normalsize

denote the expected costs of choosing the on-street (resp. private) parking space option when all other drivers play the mixed-action $p$, while

\vspace{-8pt}
\small
\begin{equation}\label{eq:cost_mix_sym}
 c_{i}^{N}(p,p)=p_{pub}\cdot c_{i}^{N}(public,p)+p_{priv}\cdot c_{i}^{N}(private,p)
\end{equation}
\normalsize

is the cost of the symmetric profile where everyone plays the mixed-action $p$.
With these at hand, we can now postulate the following Theorem.

\begin{theorem}\label{thm:sme}
The parking spot selection game $\Gamma(N)$ has a unique symmetric mixed-action Nash equilibrium $p^{NE}=(p^{NE}_{pub},p^{NE}_{priv})$, where:
\begin{itemize}
  \item $p^{NE}_{pub}=1$, \text{ if } $N\leq\sigma_{0}$ and
  \item $p^{NE}_{pub}=\frac{\sigma_{0}}{N}$, \text{ if } $N>\sigma_{0}$,
\end{itemize}

%\begin{itemize}
%  \item $p^{NE}_{pub}=1$, if $N\leq\sigma_{0}$,
%  \item $p^{NE}_{pub}=\frac{\sigma_{0}}{N}$, if $N>\sigma_{0}$,
  %with $p_{pub,sym,eq}=\frac{\lfloor\sigma_{0}\rfloor}{N}$
%\end{itemize}
%\hspace{5pt}
where $p^{NE}_{pub}=1-p^{NE}_{priv}$ and $\sigma_{0}\in\mathbb{R}$.%=\frac{R(\gamma-1)}{\delta}\in\mathbb{R}$.
\end{theorem}

\begin{proof}
The proof is given in the Appendix.
\end{proof}

\textbf{Asymmetric mixed-action equilibria:}
In the analysis of pure equilibria in Section \ref{sec:pe}, we showed that there are multiple asymmetric pure equilibria, when the number of drivers exceeds the $\sigma_0$. In general, the derivation of results for asymmetric mixed-action equilibria is much harder than for either their pure or their symmetric counterparts since the search space is much larger. Moreover, asymmetric mixed-equilibria have two more undesirable properties: a) they do not treat all players equally, \ie different players end up with \emph{a-priori} worse chances to come up with a cheap parking spot; b) their realization in practical situations is a much more difficult than their symmetric counterparts.

Therefore, in this and subsequent Sections, we base our analysis and discussion on symmetric equilibria and their (in)efficiency.

\section{Incomplete knowledge of the parking demand}\label{sec:incinfo}

The availability of complete information about the drivers' (\ie players') population is a fairly strong and unrealistic assumption.
%However, this can be a fairly unrealistic assumption, as it posits continuous scalable control message flow within autonomic network complexes.
In this Section we relax it by studying two game variants with incomplete information, where the players either share common probabilistic information about the overall parking demand or are totally uncertain about it. Note that the parking service operator may, depending on the network and information sensing infrastructure at her disposal, provide the competing drivers with different amounts of information about the demand for parking space (\eg historical statistical data about the utilization of on-street parking space).

\subsection{Probabilistic knowledge of parking demand}\label{Distributed_model_pinfo}
In the \emph{Bayesian} model of the game, the drivers determine their actions on the basis of private information, their \emph{types}. The type in this game is a binary variable indicating whether a driver is in search of parking space (\emph{active} player). Every driver knows her own type along with the strategy space, the cost functions, and the possible types of all others. However, she ignores the real state of the game at a particular moment in time, as expressed by the types of the other players, and, hence, she cannot deterministically reason out the actions being played. Instead, she draws on common \emph{prior} probabilistic information about the activity of drivers to derive estimates about the expected cost of her actions.

Formally, the Bayesian parking spot selection game is defined as follows:

\begin{definition}\label{def:bayes_parking_game}
A \emph{Bayesian Parking Spot Selection Game} is a tuple \\ $\Gamma_{B}(N)=(\mathcal{N},\mathcal{R},(w_{j})_{j\in(pub,priv)},(A_{i})_{i\in \mathcal{N}},(\Theta_{i})_{i\in \mathcal{N}},f_{\Theta})$, comprising:
\begin{itemize}
\item $\mathcal{N}$ and $\mathcal{R}$, as defined for $\Gamma(N)$,
    %$where $\mathcal{R}_{pub}$ is the set of public spots, with $R=|\mathcal{R}_{pub}|\geq1$ and $\mathcal{R}_{priv}$ the %set of private spots, with $|\mathcal{R}_{priv}|\geq N$,
%\item $A_{i}=\{public,private,\oslash\}$, the common set of possible actions for each driver $i\in \mathcal{\bar{N}}$,
\item $A_{i}=\{public,private,\oslash\}$, the set of potential actions for each driver $i\in \mathcal{N}$,
\item $\Theta_{i}=\{0,1\}$, the set of types for each driver $i\in \mathcal{N}$, where $1$ stands for active and $0$ for inactive drivers,
\item $S_{i}: \Theta_{i}\rightarrow A_{i}$, the set of possible strategies for each driver $i\in \mathcal{N}$,
\item $c_{i}^{N_B}(s(\vartheta),\vartheta)$, the cost functions for each driver $i\in \mathcal{N}$, for every type profile $\vartheta\in \times_{k=1}^{N}\Theta_{k}$ and strategy profile $s(\vartheta)\in \times_{k=1}^{N}S_{k}$,
\item $f_{\Theta}$, the prior joint probability distribution of the drivers' activity.
\end{itemize}
\end{definition}

%A set of constraints underlies the strategy function for each player.
In $\Gamma_{B}(N)$, all inactive drivers abstain from the game interaction; hence, $s_{i}(\vartheta_{i}=0)=\oslash$. On the contrary, $s_{i}(\vartheta_{i}=1) \in \{public,private\}$, with the active players also randomizing over this subset of $A_i$ choosing mixed-actions. The game is symmetric when, besides the action set, drivers share the same activity distribution. The real number of active players upon each time depends on their types and is given by $n_{act}=\sum_{k}\vartheta_{k}$.
% randomize over the subset $\{public,private\}$ (ref. Section \ref{Parking_spot_selection game}).
%In the strategic games, the cost function for each player is defined on the action profile being played. In the Bayesian games
The action profile is the effect of players' strategies on their types and is noted as $a=(s(\vartheta),\vartheta)\in \times_{k=1}^{N}A_{k}$. The cost $c_{i}^{N_B}(s(\vartheta),\vartheta)$ for the active driver $i$ under the type profile $\vartheta$ and the strategy profile $s(\vartheta)$ is

%\small
%\begin{eqnarray}
%c_{i}^{\bar{N}}(s(\vartheta),\vartheta)= \nonumber \\
%c_{i}^{\bar{N}}(s_{i}(\vartheta),s_{-i}(\vartheta),\vartheta_{i},\vartheta_{-i})= \nonumber \\
%c_{i}^{\sum_{k}\vartheta_{k}}(a_{i},a_{-i}) \nonumber \\
%c_{i}^{\sum_{k}\vartheta_{k}}(a_{i},a_{-i'})
%\label{eq:bayesian_cost}
%\end{eqnarray}

\small
\begin{equation}
c_{i}^{N_B}(s(\vartheta),\vartheta)= c_{i}^{N_B}(s_{i}(\vartheta_{i}),s_{-i}(\vartheta_{-i}),\vartheta_{i},\vartheta_{-i})
\label{eq:bayesian_cost}
\end{equation}
\vspace{0.5pt}
\normalsize

%Therefore, every active player pays according to the formulas in (\ref{equ:cost}), with $N=\sum_{k}\vartheta_{k}$, where $a_{-i'}$ is an abbreviation for the actions of all active players but player $i$. The identity of active players remains a matter of indifference due to the symmetry of the game (\ie all active players share the same action set). The real number of active players $N$ depends on the types the players own. The players access this knowledge only probabilistically. Thus, unlike the games with complete knowledge, in the bayesian implementation the active players determine their strategy exploiting the probabilistic information for the actual volume of active players.

\textbf{Equilibria:} For the Bayesian parking spot selection game, the strategy profile %$s'=(s_{1}(1),...,s_{N}(1))$
$s'\in\times_{k=1}^{N}S_{k}(\vartheta_{k}=1)$ is a Bayesian Nash equilibrium if for all $i\in\mathcal{N}$ with $\vartheta_{i}=1$:
\small
\begin{equation}\label{equ:Bpe}
s_{i}(\vartheta_{i})\in \arg \min_{s'_{i}\in S_{i}}\sum_{\vartheta_{-i}}f_{\Theta}(\vartheta_{-i}/\vartheta_{i})c_{i}^{\sum_{k}\vartheta_{k}}(s_{i}',s_{-i}(\vartheta_{-i}),\vartheta_{i},\vartheta_{-i})
\end{equation}
\normalsize
where $c_{i}^k(s_{i}',s_{-i})$, with $s_{l}(\vartheta_{l}=0)=private$, $\forall l\neq i$, is the cost of driver $i$ under profile $s$ in the game $\Gamma(k)$ and $f_{\Theta}(\vartheta_{-i}/\vartheta_{i})$ the posterior conditional probability of the active drivers \emph{given that} user $i$ is active, as derived from the application of the Bayesian rule. Therefore, $s'$ minimizes the expected cost over all possible combinations of the other drivers' types and strategies so that no active player can further lower its expected cost by unilaterally changing her strategy.

\begin{theorem}\label{thm:sme_bayes}
The Bayesian parking spot selection game $\Gamma_{B}(N)$ has unique symmetric equilibrium profiles $p^{NE_{B}}=(p^{NE_{B}}_{pub},p^{NE_{B}}_{priv})$. More specifically:
\begin{itemize}
  \item a unique (Bayesian-Nash) pure equilibrium with $p^{NE_{B}}_{pub}=1$, if $p_{act}<\frac{\sigma_{0}}{\bar{N}}$,
  \item a unique symmetric mixed-action Bayesian Nash equilibrium with  $p^{NE_{B}}_{pub}=\frac{\sigma_{0}}{Np_{act}}$, if $p_{act}\geq\min(\frac{\sigma_{0}}{N},1)$,
\end{itemize}
where $p^{NE_{B}}_{priv}=1-p^{NE_{B}}_{pub}$ and $\sigma_{0}\in\mathbb{R}$.%=\frac{R(\gamma-1)}{\delta}\in\mathbb{R}$.
\end{theorem}

\begin{proof}
We present the proof in the Appendix.
\end{proof}

\subsection{Strictly incomplete information about parking demand}\label{Distributed_model_ninfo}

The worst-case scenario with respect to the information drivers avail for making their decisions is represented by the \emph{pre-Bayesian} game variant.
%analysis of the parking spot selection game over unknown information is realized through the notion of pre-Bayesian games. In %these games, no player has any probabilistic information about the real state of others. Hence, all players ignore others' %private type or equivalently their valid action set.
In the pre-Bayesian parking spot selection game, the drivers may avail some knowledge about the upper limit of the vehicles that are \emph{potential} competitors for parking resources, yet their actual number is not known, not even probabilistically.

%In essence, the pre-Bayesian games differ from their Bayesian counterpart on the - commonly known - probability distribution %about the private types that the players exploit for their decisional operations.
Pre-Bayesian games do not necessarily have \emph{ex-post} Nash equilibria, even in mixed actions. On the other hand, all quasi-concave pre-Bayesian games \emph{do} have at least one mixed-strategy \emph{safety-level equilibrium} \cite{Ashlagi06}. %\footnote{The result holds for quasi-concave games.}
%as long as they satisfy some requirements
In the safety-level equilibrium, every player minimizes over her strategy set $S_i$ the worst-case (maximum) expected cost she may suffer over all possible types and actions of her competitors $(S_{-i},\Theta_{-i})$.

%For this type of games new classes of equilibrium concepts have been introduced. In particular, the \emph{ex-post Nash equilibrium} consists of strategies that, for every joint type profile, result in actions which are in Nash equilibrium in the corresponding strategic game []. Furthermore, a profile of strategies is an \emph{minimax regret equilibrium} if for every type profile, every player chooses the action that minimizes the maximum difference in cost with every other action (that his type allows) []. Another one is the \emph{safety-level equilibrium}, where for every type profile, every player chooses the action (that his type allows) that maximizes the minimum respective cost incurred []. The proof of their (in)existence, all the more so their computation, is not a trivial task. However, some known results can be applied in the pre-Bayesian variant of the parking spot selection game in order to extract insights to the stable conditions. In particular, the authors in [] prove that in every resource selection game with non-decreasing cost functions (namely, as in case of the parking spot selection game) the respective pre-Bayesian game has a unique symmetric safety-level equilibrium in which the players perform the unique symmetric equilibrium action of the primary strategic game.

The result of interest for our pre-Bayesian variant of the parking spot selection model $\Gamma_{pB}(N)$ is due to \cite{Ashlagi06}.
\begin{proposition}\label{prop:pre_bayes}
An action profile $a$ is the unique symmetric mixed-action safety-level equilibrium of the pre-Bayesian parking spot selection game, $\Gamma_{pB}(N)$, with non-decreasing resource cost functions, iff $a$ is the unique symmetric mixed-action equilibrium of the respective strategic game with deterministic knowledge of the number of players, $\Gamma(N)$.
\end{proposition}

We discuss the implications of this result for the efficiency of the equilibria behaviors of the drivers in Section \ref{sec:results}.

%\begin{proposition}\label{prop:pre_bayes}
%The minimal social cost in the parking spot selection game $\Gamma(N)$ is attained at the symmetric mixed action of the pre-Bayesian game $\Gamma_{pB}(\frac{N(\gamma-1)}{\delta})$.
%\end{proposition}
%
%\begin{proof}
%
%\end{proof}

%The minimal social cost in the parking spot selection game $\Gamma(N)$ is attained at the symmetric mixed action of the %pre-Bayesian game $\Gamma_{pB}(\frac{N(\gamma-1)}{\delta})$. 
\section{Numerical results}\label{sec:results}
The analysis in Sections \ref{sec:cinfo} and \ref{sec:incinfo} suggests that three important factors affect the (in)efficiency of the game equilibrium profiles. The first two are the charging policy for on-street and private parking space and their relative location, which determines the overhead parameter $\delta$ of failed attempts for on-street parking space. The third factor is the information available to the drivers when playing the game. In the following, we illustrate their impact on the game outcome and discuss their implications for real systems.

For the numerical results we adopt per-time unit normalized values used in the typical municipal parking systems in big European cities \cite{cityeu}. The parking fee for public space is set to $c_{pub,s}=1$ unit
%\euro~ for $60$-minute period,
whereas the cost of private parking space $\beta$ ranges in $(1,16]$ units and the excess cost $\delta$ in $[1, 5]$ units.
%We analyze the dynamics of the system under low to medium parking demand levels and limited public parking supply, during the time window over which the parking requests are issued, so that the capacity of current private parking services in the area suffices to fulfil all parking requests.
We consider various parking demand levels assuming that private parking facilities in the area suffice to fulfil all parking requests.

\subsection{Impact of charging policy}

%%%%mobicom
%\begin{figure}[t]
%%\vspace{-15pt}
%\begin{center}
%\begin{tabular}{cc}
% \hspace{-10pt} \includegraphics[scale=0.3]{social_cost_pure_50}
%  %\label{fig:sc_pure}
%  &
%  \hspace{-17pt} \includegraphics[scale=0.3]{social_cost_mixed_50}
%  %\label{fig:sc_mixed}
%\\
%
%\begin{scriptsize} \hspace{-3pt} a. Pure-action profiles \end{scriptsize}  & \hspace{-25pt}
%\begin{scriptsize} b. Symmetric mixed-action profiles  \end{scriptsize}
%\end{tabular}
%\end{center}
%\vspace{-15pt}
%\caption{Social cost for $N=500$ drivers when exactly $\sigma_{pub}$ drivers compete (a) or when all drivers decide to compete with probability $p_{pub}$ (b), for $R=50$ public parking spots, under different charging policies.\label{fig:sc_complete}}
%\vspace{-6pt}
%\end{figure}

\begin{figure}[t]
\vspace{-100pt}
\begin{center}
\begin{tabular}{cc}
 \hspace{-40pt}  \vspace{-70pt} 
 \includegraphics[scale=0.3]{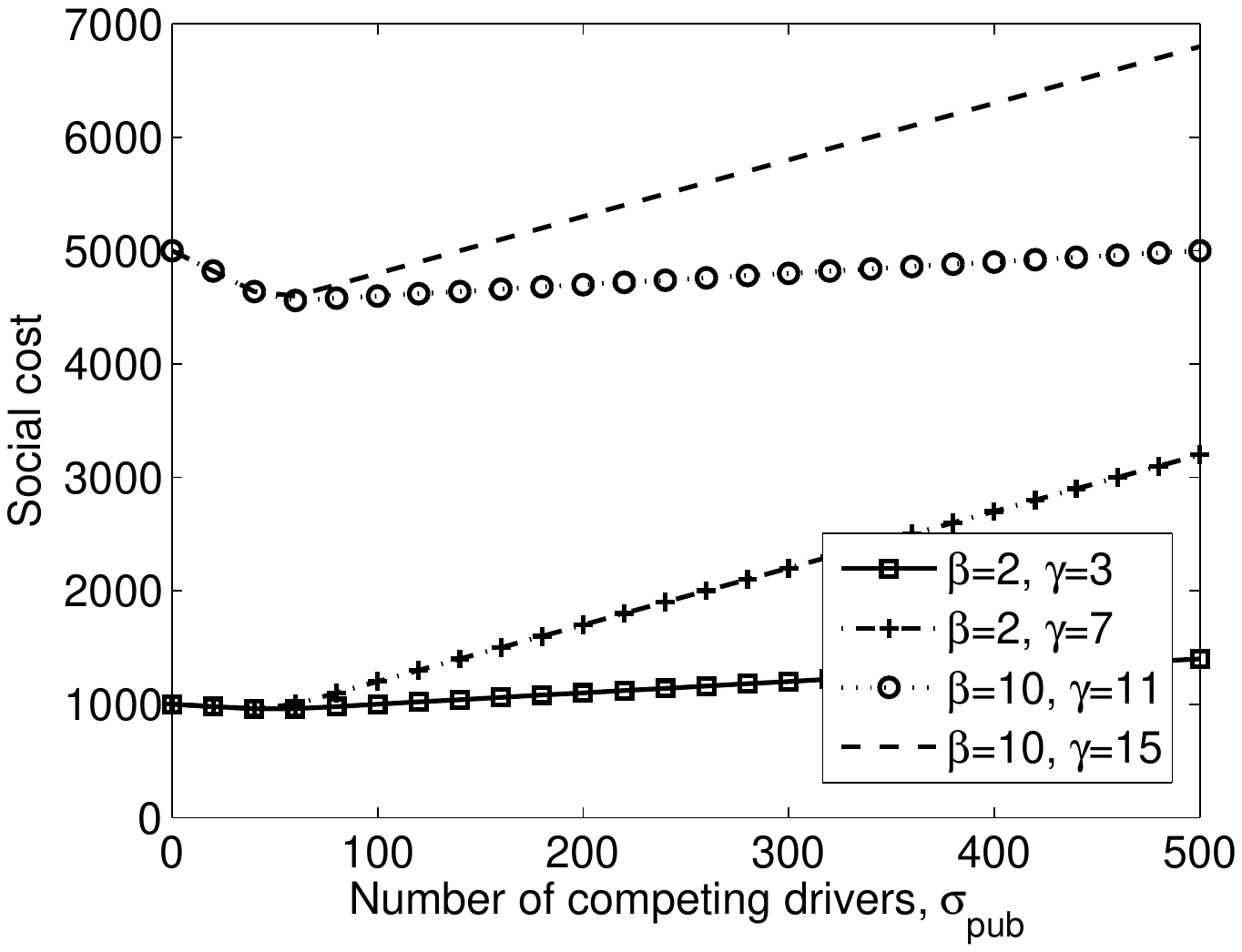}
  \hspace{-60pt} 
  \includegraphics[scale=0.3]{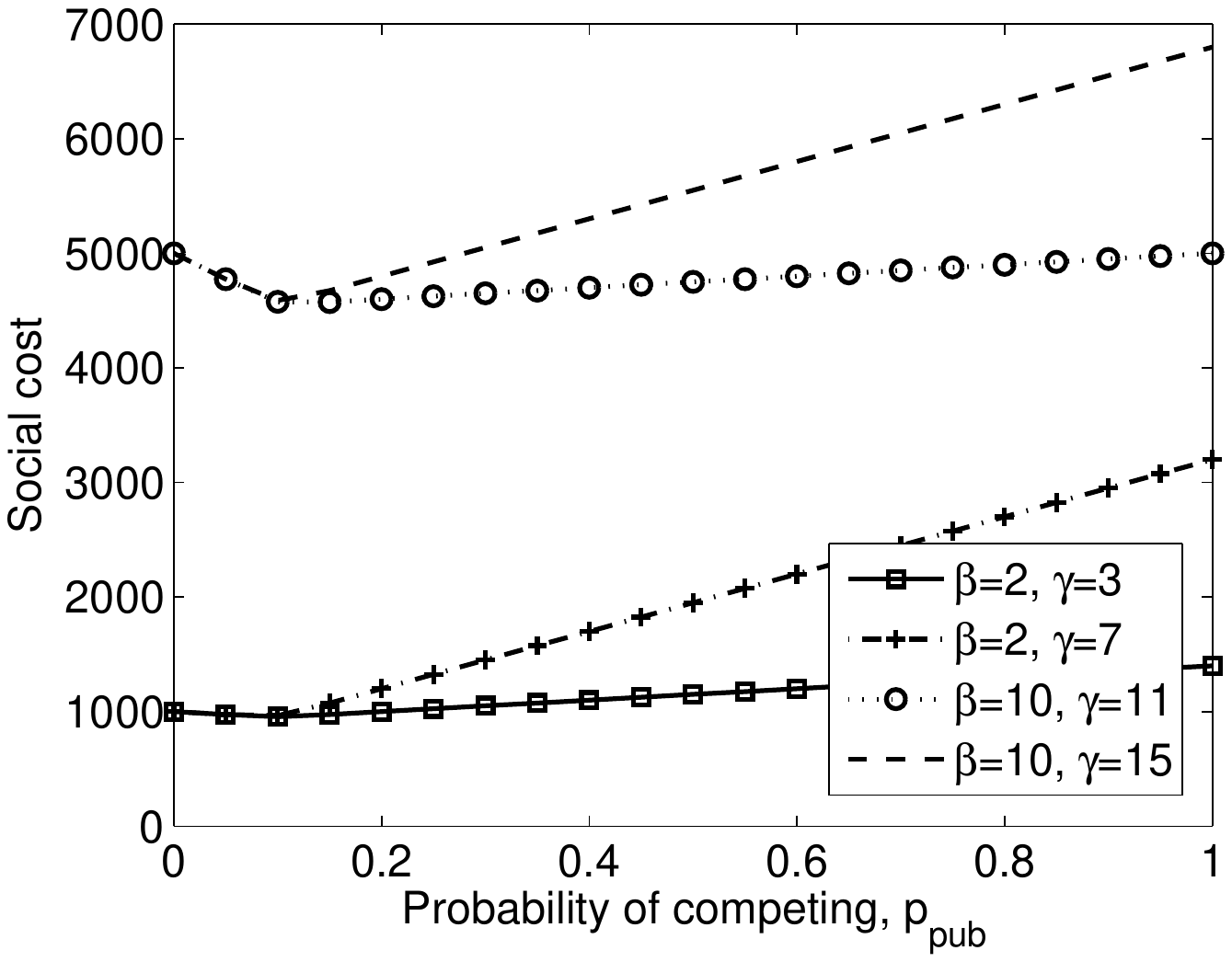}
\\

\begin{scriptsize} \hspace{-8pt} a. Pure-action profiles ~~~ b. Symmetric mixed-action profiles  \end{scriptsize}
\end{tabular}
\end{center}
\vspace{-10pt}
\caption{Social cost for $N=500$ drivers when exactly $\sigma_{pub}$ drivers compete (a) or when all drivers decide to compete with probability $p_{pub}$ (b), for $R=50$ public parking spots, under different charging policies.\label{fig:sc_complete}}
\vspace{-10pt}
\end{figure}

%%%mobicom
%\begin{figure*}[t]
%\vspace{-15pt}
%\begin{center}
%\begin{tabular}{cc}
%\hspace{-15pt}
%  \includegraphics[scale=0.35]{poa_1}
%  \includegraphics[scale=0.48]{poa_2}
%  \includegraphics[scale=0.36]{poa_3}
%\\

%\begin{scriptsize} a. $2D~ \verb"PoA"(\beta), R=160$ ~~~~~~~~~~~~~~~~~~~~~~~~~~~~~~~~~~~~ b. $3D~ \verb"PoA"(\beta,\delta), R=160$ ~~~~~~~~~~~~~~~~~~~~~~~~~~~~~~~~~~
% c. $3D~ \verb"PoA"(\beta,\delta), R=50$  \end{scriptsize}
%\end{tabular}
%\end{center}
%\vspace{-15pt}
%\caption{Price of Anarchy for $N=500$ and varying $R$, under different charging policies.\label{fig:poa_complete}}
%\vspace{-8pt}
%\end{figure*}

\begin{figure*}[t]
\vspace{-100pt}
\begin{center}
\begin{tabular}{ccc}
\hspace{-50pt}\vspace{-95pt}
\includegraphics[scale=0.40]{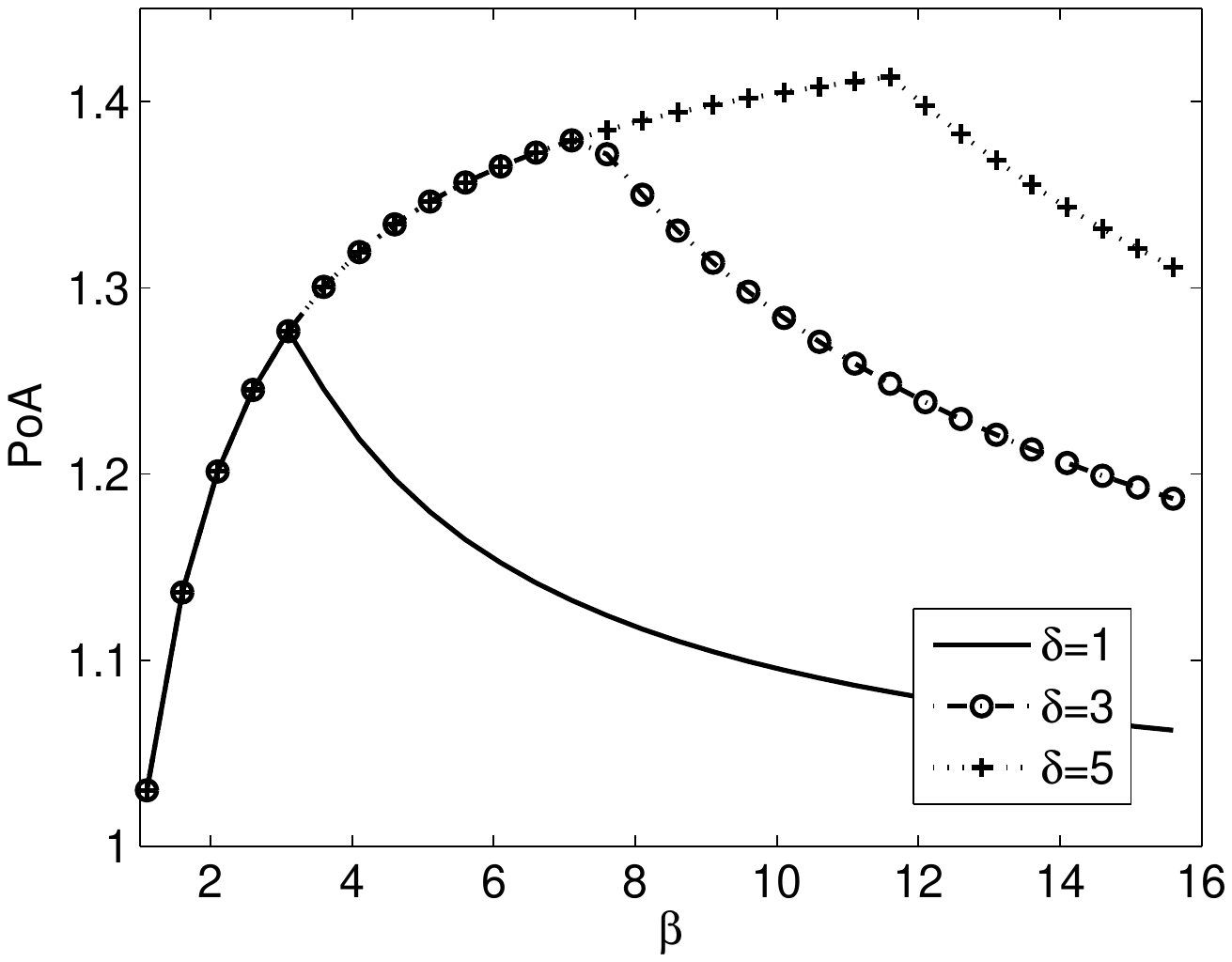}
\hspace{-75pt}
\includegraphics[scale=0.40]{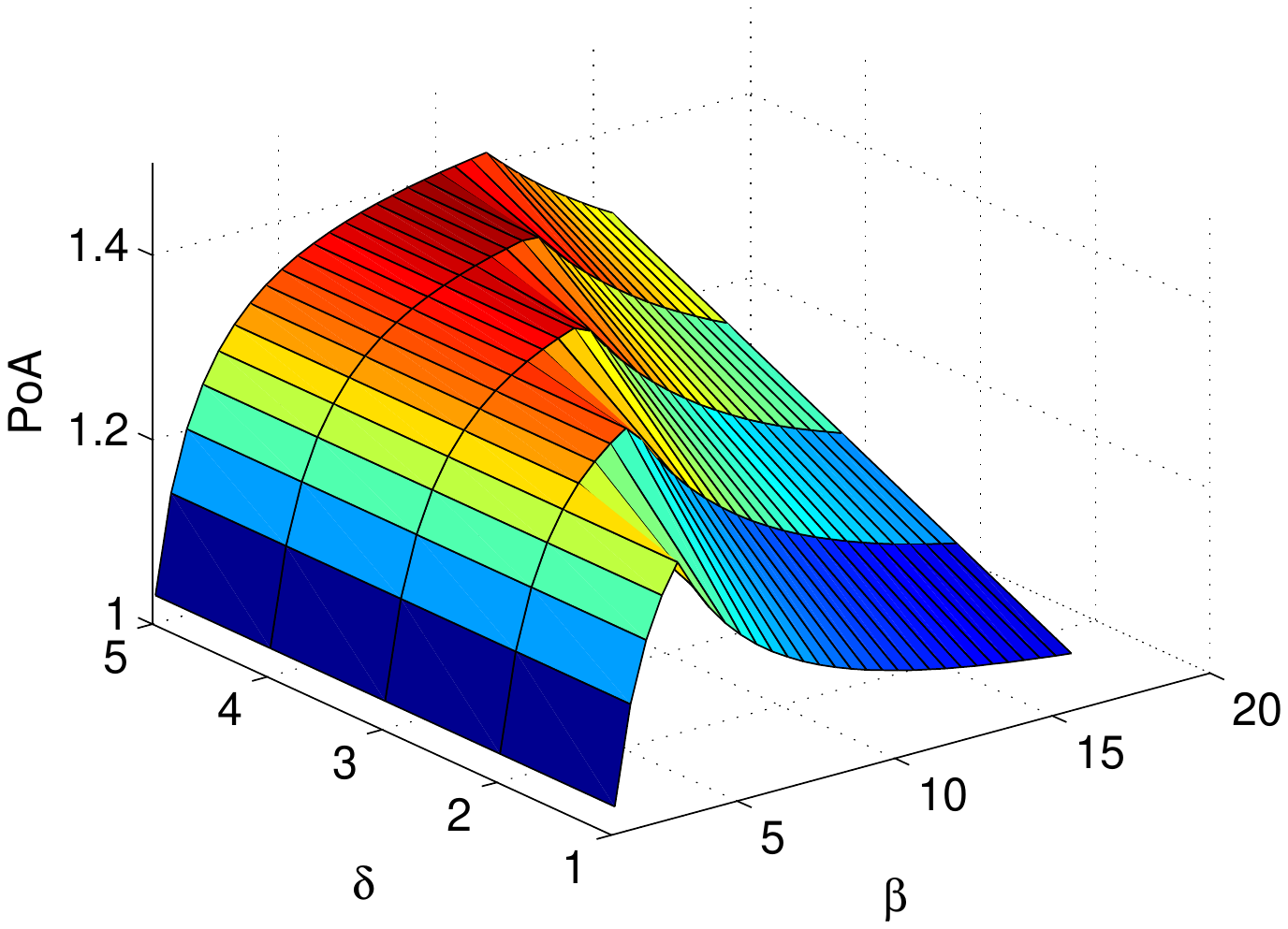}
\hspace{-75pt}
\includegraphics[scale=0.40]{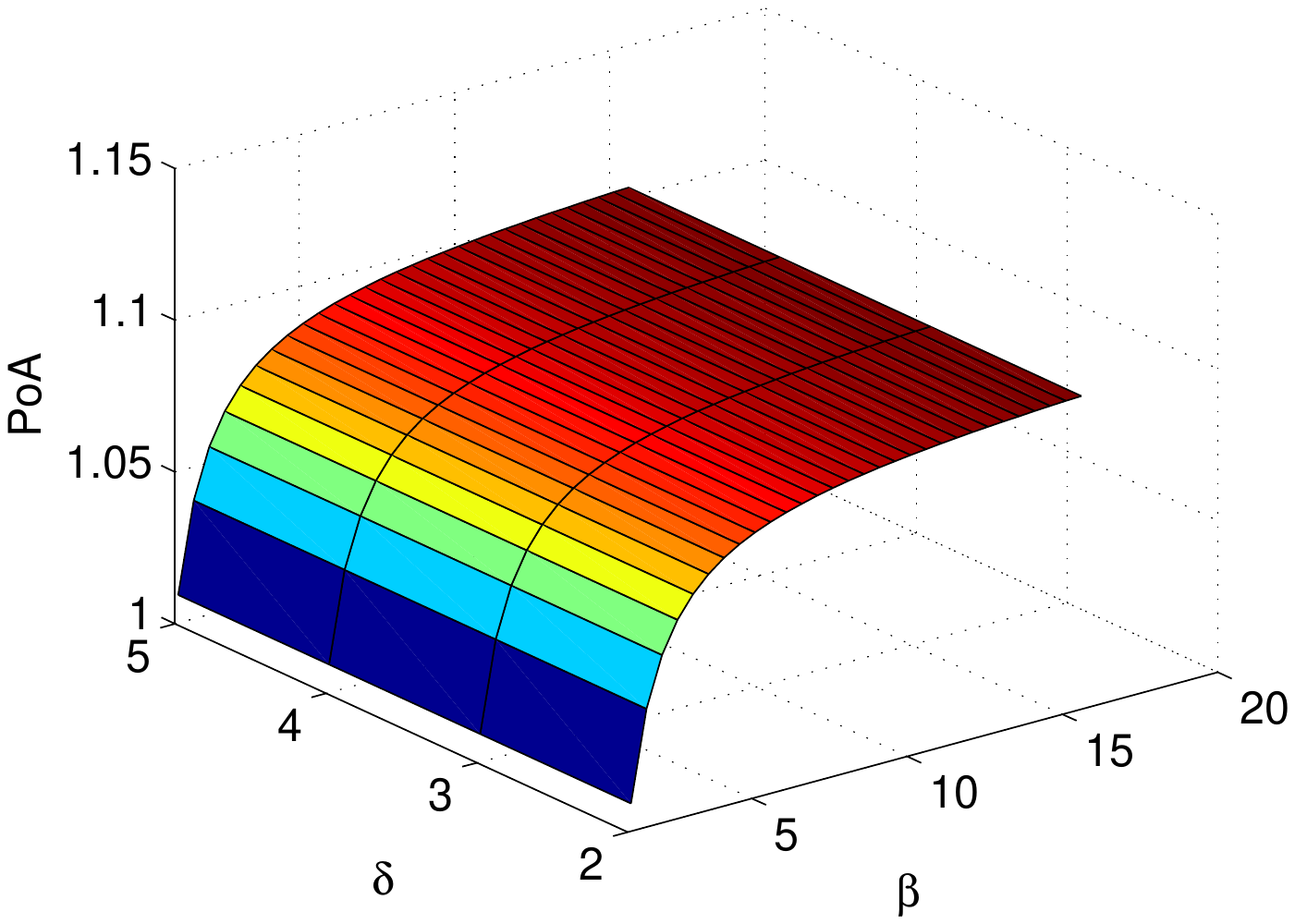}
\\
\begin{scriptsize} \hspace{-30pt} a. $2D~ \verb"PoA"(\beta), R=160$ ~~~~~~~~~~~~~~~~~~~~~~~~~~~~~~~ b. $3D~ \verb"PoA"(\beta,\delta), R=160$ ~~~~~~~~~~~~~~~~~~~~~~~~~~~~~~~~
 c. $3D~ \verb"PoA"(\beta,\delta), R=50$ \end{scriptsize}
\end{tabular}
\end{center}
\vspace{-10pt}
\caption{Price of Anarchy for $N=500$ and varying $R$, under different charging policies.\label{fig:poa_complete}}
\vspace{-10pt}
\end{figure*}

%\section{Social cost in pure-action and mixed-action profiles}\label{sec:sc_complete}
 Figure \ref{fig:sc_complete} plots the social costs $C(\sigma_{pub})$ under pure (Eq. $11$) and $C(p_{pub})$ under mixed-action strategies as a function of the number of competing drivers $\sigma_{pub}$ and competition probability $p_{pub}$, respectively, where
 %under four pairs of $(\beta,\gamma)$ values.
% Concerning the corresponding cost functions, $C(\sigma_{pub})$ is given by ($11$) and $C(p)$ by the equation
\vspace{-5pt}
 \small
 \begin{eqnarray}
  C(p) &=& c_{pub,s}\sum_{\sigma=0}^N {N \choose \sigma}p^\sigma(1-p)^{N-\sigma}\cdot \nonumber \\
  && [min(\sigma,R)+max(0,\sigma-R)\gamma+(N-\sigma)\beta]
 \end{eqnarray}
 \normalsize
 %The corresponding cost functions, $C(\sigma_{pub})$ and $C(p)$, are given by the equations
% \vspace{-12pt}
% \small
% \begin{eqnarray}
% \small
% C(\sigma_{pub})&=&\left\{
% \begin{array}{l l} c_{pub,s}\big(\sigma_{pub}+(N-\sigma_{pub})\beta\big),~~~~~~~~\text{if } \sigma_{pub}\leq R \\
%                    c_{pub,s}\big(\sigma_{pub}\delta-R(\gamma-1)+\beta N\big),~~\text{if } \sigma_{pub}> R \end{array} \right. \nonumber \\
% C(p) &=& \sum_{m=0}^N {N \choose m}p^m(1-p)^{N-m} [min(m,R)+max(m-R,0)\gamma+(N-m)\beta]c_{pub} \nonumber
% \end{eqnarray}
% \normalsize
Figure \ref{fig:sc_complete} motivates two remarks. Firstly, the social cost curves for pure and mixed-action profiles have the same shape. This comes as no surprise since for given $N$, any value for the expected number of competing players $0\leq\sigma_{pub}\leq N$ can be realized through appropriate choice of the symmetric mixed-action profile $p$. Secondly, the cost is minimized when the number of competing drivers is equal to the number of on-street parking spots. The cost rises when either competition exceeds the available on-street parking capacity or drivers are overconservative in competing for on-street parking. In both cases, the drivers pay the penalty of the lack of coordination in their decisions. The deviation from optimal grows faster with increasing price differential between the on-street and private parking space.

Whereas an optimal centralized mechanism would assign exactly $\min(N,R)$ public parking spots to $\min(N,R)$ drivers, if $N>R$, in the worst-case equilibrium the size of drivers' population that actually competes for on-street parking spots exceeds the real parking capacity by a factor $\sigma_{0}$ which is a function of $R$, $\beta$ and $\gamma$ (equivalently, $\delta$) (see Lemma \ref{le:pe}). This inefficiency is captured in the $\verb"PoA"$ plots in Figure \ref{fig:poa_complete} for $\beta$ and $\delta$ ranging in $[1.1, 16]$ and $[1,5]$, respectively. The plots illustrate the following trends:

\textbf{Fixed $\delta$ - varying $\beta$:} For $N\leq \sigma_{0}$ or, equivalently, for $\beta\geq\frac{\delta(N-R)+R}{R}$, %from Proposition (\ref{prop:poa}),
it holds that
%\begin{equation}\label{eq:der_poa_beta_1}
$\frac{\vartheta PoA}{\vartheta\beta}<0$
%\end{equation}
and therefore, the $\verb"PoA"$ is strictly decreasing in $\beta$. %Hence the social cost converges to the optimum as $\beta\rightarrow\infty$, practically impossible.
On the contrary, for $\beta<\frac{\delta(N-R)+R}{R}$, the $\verb"PoA"$ is strictly increasing in $\beta$, since
%\begin{equation}\label{eq:der_poa_beta_2}
$\frac{\vartheta PoA}{\vartheta\beta}>0$.
%\end{equation}
%Hence, as the cost $\beta\rightarrow 1$, the  $\verb"PoA"\rightarrow 1$.

\textbf{Fixed $\beta$ - varying $\delta$:} For $N\leq \sigma_{0}$ or, equivalently, for $\delta\leq\frac{R(\beta-1)}{N-R}$ we get
%\begin{equation}\label{eq:der_poa_delta_1}
$\frac{\vartheta PoA}{\vartheta\delta}>0$.
%\end{equation}
Therefore, the $\verb"PoA"$ is strictly increasing in $\delta$. %and hence the social cost converges to the optimal as the distance at which the private parking facility is located tends to infinity.
For $\delta>\frac{R(\beta-1)}{N-R}$, we get
%\begin{equation}\label{eq:der_poa_delta_2}
$\frac{\vartheta PoA}{\vartheta\delta}=0$.
%\end{equation}
Hence, if $\delta$
%the distance between the areas of public and private parking spots
exceeds $\frac{R(\beta-1)}{N-R}$, $\verb"PoA"$ is insensitive to changes of the excess cost $\delta$.
%
%%we can upper bound the social cost and hence the $\verb"PoA"$, by $1+\varepsilon$ with $\varepsilon>0$:
%%\begin{equation}\label{eq:epsilon_beta_1}
%%\verb"PoA"=\frac{\gamma N-R(\gamma-1)}{R+\beta(N-R)}=1+\varepsilon,\Rightarrow
%%\beta=\frac{\delta(N-R)-\varepsilon R}{\varepsilon(N-R)}
%%\end{equation}
%%
%%Similarly, for $\beta<\frac{\delta(N-R)+R}{R}$, it holds that :
%%\begin{equation}\label{eq:epsilon_beta_2}
%%\verb"PoA"=\frac{\sigma_{0}\delta-R(\gamma-1)+\beta N}{R+\beta(N-R)}
%%\= 1+\varepsilon,\varepsilon>0\Rightarrow
%%\beta=\frac{R(1+\varepsilon)}{R-\varepsilon(N-R)}
%%\end{equation}
%%
%%Furthermore, the cost relation of the two parking options requires $\beta>1$. Hence
%%
%%\begin{equation}\label{eq:epsilon_1}
%%\varepsilon<\frac{\delta(N-R)}{N}
%%\end{equation}
%%
%%However, $\beta$ is a positive real number. Hence,
%%\begin{equation}\label{eq:epsilon_2}
%%\varepsilon<\frac{R}{N-R}
%%\end{equation}
%%
%%which follows Proposition \ref{prop:bound_poa}.
%
%%for a given travel distance between the two parking options (\ie parameter $\delta$), increasing the cost for the private parking spots (\ie parameter $\beta$), may repel drivers from this option, rendering the distributed approach significantly suboptimal.
%%
%\vspace{-0.05cm}

Practically, the equilibrium strategy emerging from the current-practice parking search behavior, approximates the optimal coordinated mechanism when the operation of private parking facilities accounts for drivers' preferences as well as estimates of the typical parking demand and supply. More specifically, if, as part of the pricing policy, the cost of private parking is less than $\frac{\delta(N-R)+R}{R}$ times the cost of on-street parking, then the social cost in the equilibrium profile approximates the optimal social cost as the price differential between public and private parking decreases.
% the more optimum the social cost becomes.
This result is inline with the statement in \cite{Larson10}, arguing that ``price differentials between on-street and off-street parking should be reduced in order to reduce traffic congestion''.  Note that the $\verb"PoA"$ metric also decreases monotonically for high values of the private parking cost when the private parking operator desires to gain more than $\frac{\delta(N-R)+R}{R}$ times the cost of on-street parking towards a bound that depends on the excess cost $\delta$. Nevertheless, these operating points correspond to high absolute social cost, \ie the minimum achievable social cost is already unfavorable due to the high fee paid by $N-R$ drivers that use the private parking space (see Fig. \ref{fig:sc_complete}).
%in the cost for private parking, significantly beyond the extra cost $\delta$, optimizes $\verb"PoA"$, it actually harms the social welfare, when %talking in absolute terms (see Fig. \ref{fig:sc_complete}).
On the other hand, there are instances, as in case of $R=50$ (see Fig. \ref{fig:poa_complete}), where the value $\frac{\delta(N-R)+R}{R}$ consists a non-realistic option for the cost of private parking space, already for $\delta>1$. Thus, contrary to the previous case, $\verb"PoA"$ only improves as the cost for private parking decreases. Finally, for given cost of the private parking space, the social cost can be optimized by locating the private facility in the proximity of the on-street parking spots so that the additional travel distance is reduced and the excess cost remains below $\frac{R(\beta-1)}{N-R}$.

%In other words, the equilibrium strategy emerging from the current-practice parking search behavior, can approximate the optimal coordinated mechanism
%% uncoordinated scheme can reach the advantageous coordinated mechanism
%under a smart operation of the private space that accounts for public's preferences. More specifically, by locating the private parking facilities in the proximity of the public parking area (\ie low $\delta$ values) and adapting properly the price differentials between public and private parking (\ie low $\beta$ values). 

\subsection{Impact of information about competition}

%\vspace{-5pt}
Looking at the mixed-action equilibria,
%The randomization over the pure action set implies players' inclination on either of the two parking options. For the strategic game,
Theorem \ref{thm:sme} indicates that drivers' intention to compete for public parking resources is shaped by the charging policy, the number of players and the public parking capacity. Indeed, players start to withdraw from competition as competition intensity rises over the threshold $\sigma_{0}=\frac{R(\gamma-1)}{\delta}$. For the Bayesian implementation, the rationale behind active players' behavior is almost the same. The only difference is that the players adjust their strategies on estimations for the competition level, based on the commonly known probabilistic information. Therefore, the probability to compete decreases with the expected number of competitors $Np_{act}$, if this number exceeds the threshold $\sigma_{0}$ of the strategic games (see Theorem \ref{thm:sme_bayes}). Furthermore, for both game formulations, players start to renege from competition as the distance between public and private parking facilities (\ie $\delta$) is extended or the number of opportunities in public parking decreases (\ie $R$) or the price for private parking reservation drops (\ie $\beta$). Figure \ref{fig:p_act} depicts the effect of these parameters on the equilibrium mixed-action, for strategic (\ie $p_{act}=1$) and Bayesian games (\ie $p_{act}\in\{0.5,0.7\}$).

%%%mobicom
%\begin{figure}[t]
%%\vspace{-15pt}
%\begin{center}
%\begin{tabular}{cc}
% \hspace{-10pt} \includegraphics[scale=0.3]{p_act1}
% \hspace{-15pt} \includegraphics[scale=0.3]{p_act2}
%  \end{tabular}
%\end{center}
%\vspace{-15pt}
%\caption{Probability of competing in equilibrium, for $R=50$. Left: Strategic and Bayesian games under fixed charging scheme $\beta=5,\gamma=7$. Right: Strategic games under various charging schemes $\beta\in\{3,6\},\gamma\in\{4,5,7,8\}$.\label{fig:p_act}}
%\vspace{-8pt}
%\end{figure}

\begin{figure}[t]
\vspace{-100pt}
\begin{center}
\begin{tabular}{cc}
 \hspace{-40pt} \includegraphics[scale=0.3]{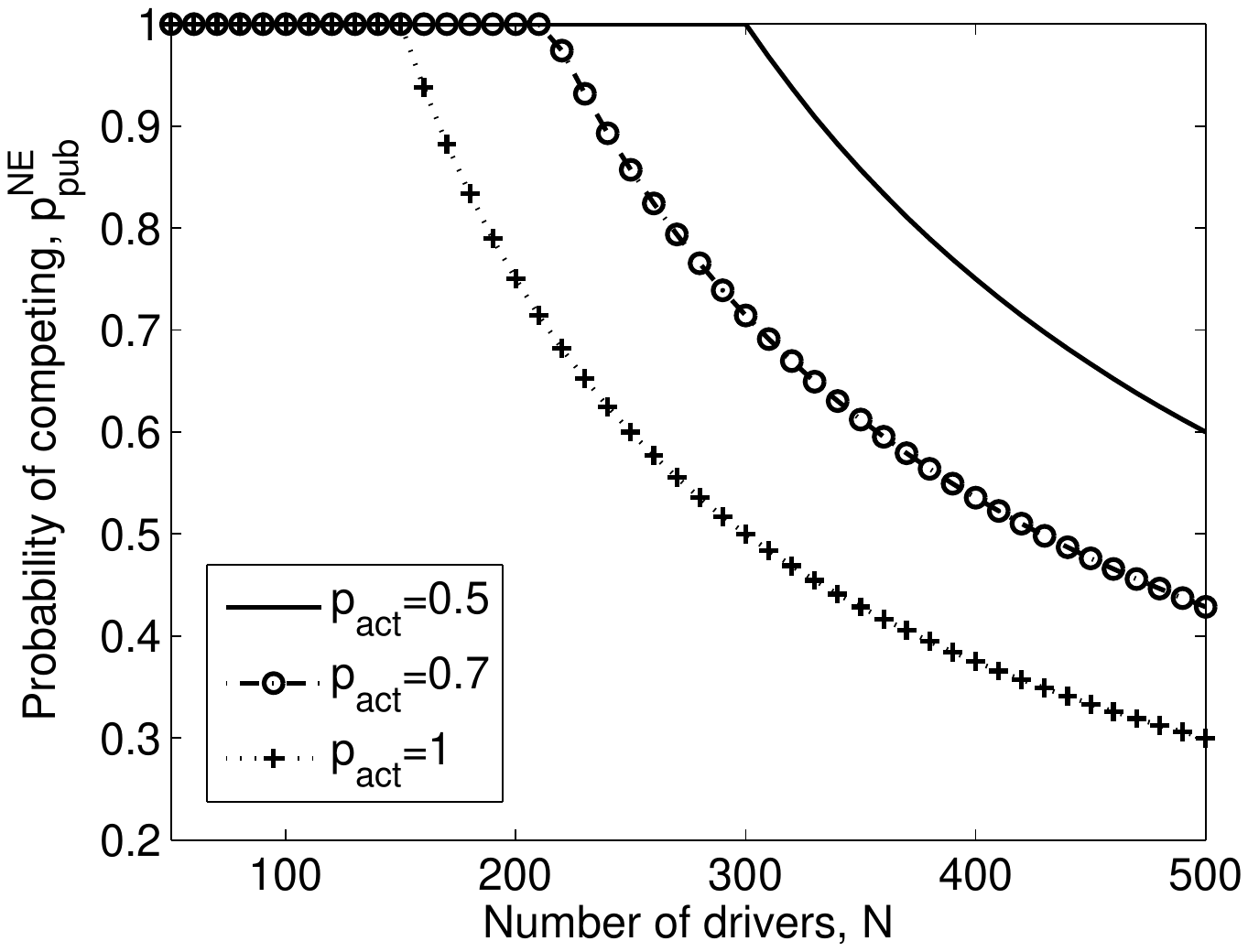}
 \hspace{-60pt} \includegraphics[scale=0.3]{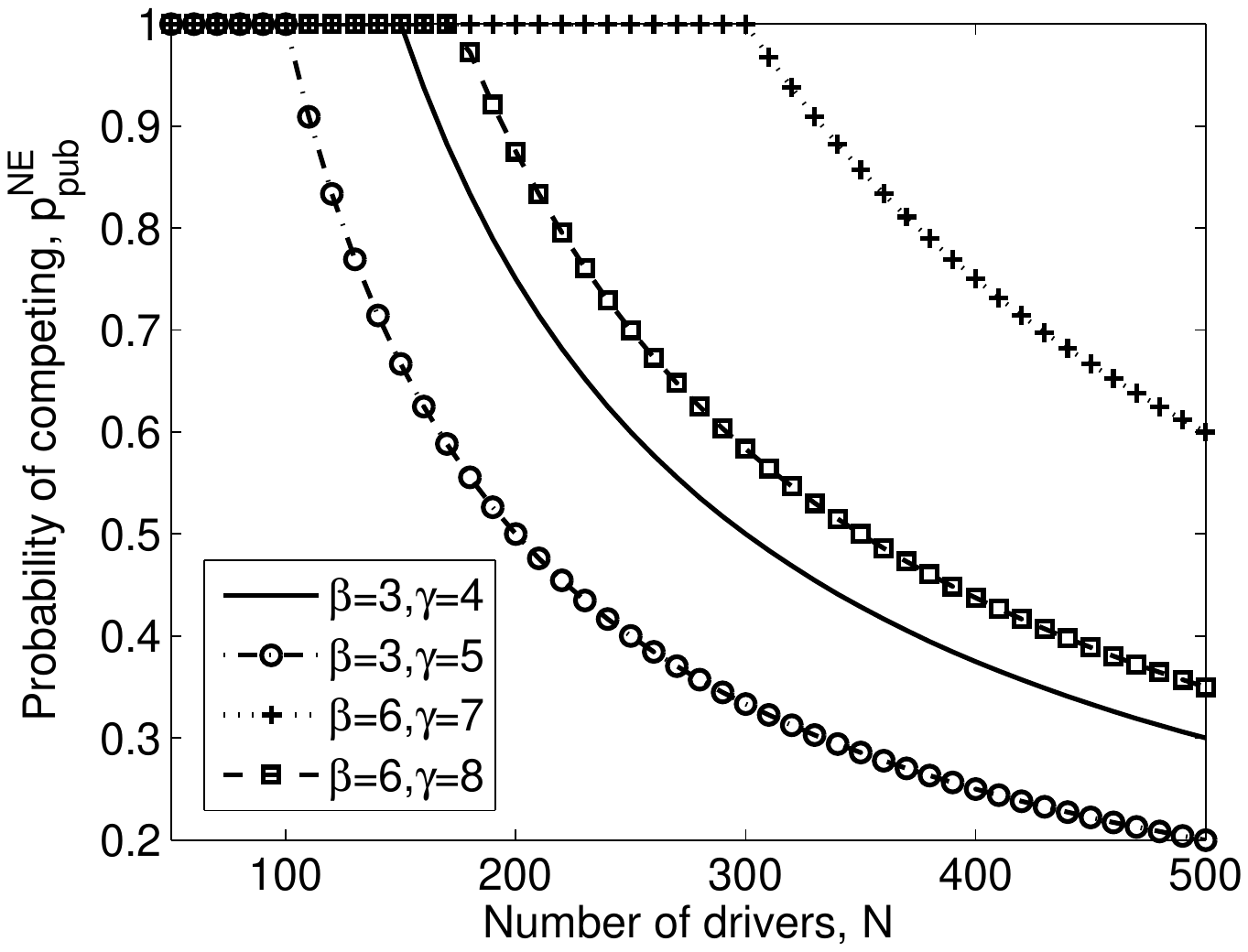}
  \end{tabular}
\end{center}
\vspace{-85pt}
\caption{Probability of competing in equilibrium, for $R=50$. Left: Strategic and Bayesian games under fixed charging scheme $\beta=5,\gamma=7$. Right: Strategic games under various charging schemes $\beta\in\{3,6\},\gamma\in\{4,5,7,8\}$.\label{fig:p_act}}
\vspace{-10pt}
\end{figure}

%\vspace{-1pt}
\textbf{Less-is-more phenomena under uncertainty:}\label{sec:less_more} \\Less intuitive are the game dynamics in its pre-Bayesian variant, when users only avail an estimate of the maximum number of drivers that are \emph{potentially} interested in parking space. From Proposition \ref{prop:pre_bayes}, the mixed-action safety-level equilibrium corresponds to the mixed action equilibrium of the strategic game $\Gamma(N)$. However, we have seen that, when the players outnumber the on-street parking capacity: a) the mixed-action equilibrium in the strategic game generates higher expected number of competitors than the optimal value $R$ (see Theorem \ref{thm:sme}); b) the social cost conditionally increases with the probability of competing (see Fig. \ref{fig:sc_complete}); c) the probability of competition decreases with the number of players $N$ (see Fig. \ref{fig:p_act}). Therefore, at the safety-level equilibrium of the game, the drivers end up randomizing the pure action $public$ with a lower probability than that corresponding to the game they actually play, with $k\leq N$ players. Hence, the resulting number of competing vehicles is smaller and, cumulatively, they may end up paying less than they would if they knew deterministically the competition they face.

%\vspace{-1pt}
One question that becomes relevant is for which (real) number $K$ of competing players do the drivers end up paying the \emph{optimal} cost. Practically, if $p_{N}^{NE}=(p_{pub,N}^{NE},\\p_{priv,N}^{NE})$ denotes the symmetric mixed-action equilibrium for $\Gamma(N)$, we are looking for the value of $K$ satisfying:

%\small
\vspace{-8pt}
\begin{equation}\label{eq:exp_comp}
Kp_{pub,N}=R \Rightarrow K =\frac{RN}{\sigma_{0}} = \frac{\delta}{\gamma-1}N \nonumber
\end{equation}
\normalsize
%Thus, by (\ref{eq:exp_comp}) and Theorem \ref{thm:sme}
%\small
%\begin{equation}
%\vspace{-7pt}
%K\frac{\sigma_{0}}{N}=R\Rightarrow K =\frac{RN}{\sigma_{0}} = \frac{\delta}{\gamma-1}N \nonumber
%\end{equation}
%\normalsize
namely, when $\frac{\delta}{\gamma-1}N$ (rounded to the nearest integer) drivers are seeking for parking space under uncertainty conditions, in the induced equilibrium they end up paying the minimum possible cost, which is better than what they would pay under complete information about the parking demand.

\section{Related work}\label{Related_work}

Various aspects of the broader parking space search problem have been addressed in the literature. The centralized systems in \cite{Boehle08} and \cite{Wang11} monitor and reserve parking places within a city and are shown to better distribute the car traffic volume. The first system consists of four components: an on-board device located in the vehicle, intelligent network enabled lampposts, a sensor network that monitors the availability of parking places and a centralized parking place scheduling/reserving service; whereas the second architecture utilizes both the Internet and Wi-Fi technology to realize the monitoring and reservation task. Likewise, the authors in \cite{Mathur10} present, design, implement, and evaluate a system that generates a real-time map of parking space availability. The map is constructed at a central server out of aggregate data about parking space occupancy, collected by vehicles circulating in the considered area. In \cite{Lu09} Lu \etal~propose SPARK for reducing the parking search delay. SPARK consists of three distinct services, \ie real-time parking navigation, intelligent antitheft protection and friendly parking information dissemination, all making use of roadside network infrastructure. On the contrary, in \cite{Verroios11} and \cite{Caliskan06}, information about the location and vacancy of parking spots is opportunistically disseminated among vehicles. In \cite{Verroios11} the vehicle nodes solve a variant of the Time-Varying Travelling Salesman problem while dynamically planning the best feasible trip along all (reported to be) vacant parking spots. The proposed method is shown via simulation results to achieve near-optimal performance, yet it makes in advance the rather debatable implication that vehicles' trip follows all reported spots. Whereas, the work in \cite{Caliskan06} uses a topology-independent scalable information dissemination algorithm and takes simulation measurements for the profile of nodes' cache entities, under various dissemination criteria.
%A bandwidth efficient protocol for disseminating parking information has been proposed in \cite{Caliskan06}. This protocol uses a topology-independent scalable information dissemination algorithm for discovering vacant parking spots. Another approach for locating available parking spots, exploiting relevant information that is opportunistically disseminated among vehicles, is presented in \cite{Verroios11}. The authors rely on the Time-Varying Travelling Salesman problem to formulate the best feasible trip along all - reported as - vacant parking spots, accounting for the computational resource constraints of vehicular nodes.

%In the absence of, possibly centrally-driven, coordination, the drivers pursue selfishly to minimize the cost of access to parking facilities. However, this common intuitive decision, combined with the scarcity of public parking capacity in urban curbside of typical center areas, give rise to \emph{tragedy of commons} effects \cite{Hardin68} and highlight the game-theoretic dynamics behind the parking spot selection problem.
Game-theoretical dimensions in general parking applications explicitly acknowledged and treated in \cite{Arnott06}, \cite{Arbat07} and \cite{Ayala11}. In \cite{Arnott06}, the games are played among parking facility providers and concern the location and capacity of their parking facility as well as which pricing structure to adopt. Whereas, in the other two works, the strategic players are the drivers. In \cite{Arbat07}, which seeks to provide cues for optimal parking lot size dimensioning, the drivers decide on the arriving time at the lot, accounting for their preferred time as well as their desire to secure a space. In a work more relevant to ours, Ayala \etal~ in \cite{Ayala11} model centralized and distributed parking spot assignment methods. The drivers exploit (or not) information on the location of others to serve their self-interest, that is, occupy an available parking spot at the minimum possible travelled distance. Finally, economic effects, this time of congestion pricing, are analyzed in \cite {Larson10} by Larson \etal, through a queueing model for drivers who circulate in search for on-street parking.

Our work approaches the parking assistance service as an instance of the more general competitive contexts introduced in Section \ref{Introduction}. Rather than proposing a particular parking assistance scheme or algorithm, as the cited papers of the first paragraph do, we draw our attention on fundamental determinants of the parking search process efficiency. We formulate three variants of the parking resource selection game (strategic, Bayesian, and pre-Bayesian) to provide normative prescriptions for the impact of information on drivers's decisions. We abstract from spatiotemporal variations of demand and supply and consider generic yet realistic pricing schemes for the service in question. Our expectation is that the obtained results may be deemed relevant to a broader class of competitive service provision scenarios\footnote{It is tempting to draw parallels with the way auctioning mechanisms provide a powerful generic abstraction for treating network resource allocation problems, (\ie spectrum sharing, online sponsored search engines) \cite{Koutsopoulos10}.}.

\section{Conclusions - Discussion}\label{sec:discussion}
%In this paper, we draw on the resource selection framework to investigate game-theoretic abstractions of the parking search process, where drivers are called to decide whether to compete for the limited low-cost public parking space or head for more expensive private parking facilities, exploiting or not information for the overall parking demand.
In this paper, we have devised game-theoretic abstractions of the parking search process. Cheaper on-street parking space and more expensive private parking facilities are modeled as discrete resources and drivers as strategic players that decide on whether to compete or not for the former,
%cheaper but scarce and competitive on-street parking space or the more expensive private parking facilities, exploiting or not
under information of variable accuracy.
%for the overall parking demand.
Our results dictate, sometimes counterintuitive, conditions under which different charging policies and information amounts for the parking demand, reduce the inefficiency of the equilibrium strategies and favor the social welfare. The parking assistance service constitutes an instance of service provision within competitive networking environments, where more information
%not only cannon secure efficient service delivery,
does not necessarily improve the efficieny of service delivery but, even worse, may hamstring users' efforts to maximize their benefit. This result, obtained under the particular \emph{full rationality} assumptions, has direct practical implications since it challenges the need for more elaborate information mechanisms and promotes certain policies for information dissemination for the service provider.
%show that revealing only partial information on parking demand may resolve better the competition between individuals.
% Results on Price of Anarchy show that although the equilibrium strategies induce higher cost than the optimal, smart charging schemes manage to defeat this inefficiency. %Furthermore, a counterintuitive result suggests that drivers' innate uncertainty for the parking demand may conditionally favor the social welfare.
%\vspace{-1pt}

In the remaining of this Section, we iterate on two implicit assumptions behind the game models we introduced in Sections \ref{sec:cinfo} and \ref{sec:incinfo}, which can motivate further research work.

\textbf{Drivers' indifference among individual parking spots:} The formulation of the parking spot selection game assumes that drivers
%prefer the cheaper option of on-street parking over private parking lots but they
do not have any preference order over the $R$ on-street parking spots. This could be the case when these $R$ spots are quite close to eachother, resulting in practically similar driving times to them and walking times from them to the drivers' ultimate destinations.

When drivers avail preferences over different parking spots, we come up with an instance of \emph{one-sided matching (assignment) games}. The objective then is to detect an assignment that no subset of the drivers could be better off if they exchanged their allocated spots with eachother. At a theoretical level, the search is for mechanisms that treat all drivers fairly, are strategy-proof, \ie the drivers are motivated to advertise their true preference orders because they cannot gain by lying about them, and efficient in some Pareto-optimality sense. The random priority and the probabilistic serial assignment are two mechanisms that compromise these requirements \cite{Moulin01}; they could be incorporated nicely in a centralized system, whereby drivers would notify the central server about their destinations and the latter would derive their ordinal (or cardinal) preference orders and make the assignments.

\textbf{Drivers' rationality:} Yet stronger and long debatable is the assumption that drivers \emph{do} behave as fully rational decision-makers. Full (or global) rationality demands that the drivers can exhaustively analyze the possible strategies available to themselves and the other drivers, identify the equilibrium profile, and take the respective actions to realize it. Simon, already more than half a century ago \cite{Simon55}, challenged both the \emph{normative} and \emph{descriptive} capacity of the fully rational decision-maker, arguing that human decisions, are most often made under time, knowledge and computational constraints and draw on simpler \emph{cognitive heuristics}. Much research work has been undertaken since then on decision-making under \emph{bounded rationality}, primarily within the cognitive psychology community, which reports experimental evidence of deviation from the global rationality directives (see, for example, \cite{Shafir95} for a survey) and/or proposes relevant heuristics, \eg \cite{Goldstein02}.

Interestingly, the conclusions from these two modeling approaches are not necessarily in conflict and our results exemplify this. Figure \ref{fig:p_act} illustrates that the symmetric equilibrium probability $p^{NE}_{pub}$ decreases as the number of competing drivers grows (see the discussion in Section \ref{sec:results}).
%Experimental evidence from the cognitive psychology community, suggesting that decision agents more generally tend to be less competitive as the number of competitors increases, has been recently reported in \cite{Garcia09} under the term \emph{N-effect}.
A similar experimental result, suggesting that decision agents more generally tend to be less competitive as the number of competitors increases, even when the chances of success remain constant, has been recently reported from the cognitive psychology community in \cite{Garcia09} under the term \emph{N-effect}. The comparison of the two decision-making modeling approaches both in the context of the parking spot selection problem and more general decision-making contexts, is an interesting area worth of further exploration.
%\vspace{-10pt} 

%\input{conclusions}
%\input{Acknowledgments}

\bibliographystyle{abbrv}
\bibliography{parking}  % sigproc.bib is the name of the Bibliography in this case
% ACM needs 'a single self-contained file'!
%

%\balancecolumns
\appendix

\section{Pure equilibria of $\Gamma(N)$ via the potential function}

The game $\Gamma(N)$ is a congestion game; thus, it accepts an exact potential function $\Phi(\cdot)$ \cite{Monderer96}. As discussed in Section \ref{sec:cinfo}, the $2^N$ different action profiles of $\Gamma(N)$ can be grouped into $N+1$ different meta-profiles $(m,N-m), 0\leq m\leq N$, where $m$ is the number of drivers that decide to compete for on-street parking. Therefore, the potential function is effectively a function of $m$ and can be written as
\vspace{-10pt}
\begin{equation}
\Phi(a)\sim \Phi(m) = \sum_{j \in \mathcal{R}}\sum_{k=0}^{n_j(a)}w_j(k)
\vspace{-5pt}
\end{equation}
where $n_j(a)$ the number of drivers using resource $j$ under action profile $a$.
Therefore, for $m\leq R$,
\vspace{-5pt}
\small
\begin{eqnarray}
\Phi(m) &=& (N-m)c_{priv} + \sum_{k=1}^m c_{pub,s} \nonumber \\
&=& c_{pub,s}[\beta N-(\beta-1) m]
\end{eqnarray}
\vspace{-5pt}
\normalsize
whereas, for $m>R$
\small
\vspace{-5pt}
\begin{eqnarray}
\Phi(m) &=& (N-m)c_{priv} + \sum_{k=1}^m min\left(1,\frac{R}{k}\right)c_{pub,s}+ \nonumber \\
&& \left[1-min\left(1,\frac{R}{k}\right)\right]c_{pub,f}  \\
&=& c_{pub,s}\left[\beta N+\delta m-R(\gamma-1)+R(1-\gamma)\cdot \sum_{k=R+1}^m\frac{1}{k}\right] \nonumber \\
&=& c_{pub,s}\left[\beta N+\delta m-R(\gamma-1)+R(1-\gamma)\cdot (H_m-H_{R+1})\right] \nonumber
\end{eqnarray}
\normalsize
$H_n=\underline{\gamma}+log(n)+O(1/n)$ is the $n^{th}$ harmonic number; and $\underline{\gamma}$ the Euler constant. The pure NE strategies coincide with the local minima of the potential function. For $m \leq R$, $\partial\Phi(m)/\partial m< 0$ and the minimum is obtained at $m$, as derived in Theorem \ref{thm:pe}.

For $m>R$, demanding $\partial\Phi(m)/\partial m=0$ we get
\begin{equation}
\delta+\frac{R(1-\gamma)}{m_{NE}}=0
\end{equation}
which yields $m_{NE}=\frac{R(\gamma-1)}{\delta}=\sigma_0$, \ie the value we got through the analysis in Section \ref{sec:cinfo}.

\section{Proof of Theorem 3.2}%\ref{thm:sme}

%\begin{proof}
The symmetric equilibrium for $N\leq \sigma_0$ corresponds to the pure NE we derived in Theorem \ref{thm:pe}.
%  %if $N\leq\frac{R(b-1)}{b-a}$, the equilibrium condition argues that all drivers risk for public parking space. Therefore, the unique - symmetric - pure equilibrium profile with $\sigma_{pub,eq}=N$ matches the unique symmetric mixed equilibrium with $p_{pub,sym,eq}=1$.
%
To compute the equilibrium for $N>\sigma_0$ we invoke the condition that equilibrium profiles must fulfil
%%For the remaining cases, that is for $N>\frac{R(b-1)}{b-a}$, the equilibrium conditions bound $\sigma_{pub}$ between $\sigma_{0}-1$ and $\sigma_{0}$ and hence restrict the probability space to $0<p_{pub,eq},p_{priv,eq}<1$, where $\sigma_{0}=\frac{R(b-1)}{b-a}\in\mathbb{R}$. By (\ref{eq:fun_mini}), we have that
%%\small
%%\begin{eqnarray}
%%q(p_{eq})= \nonumber
%%\max(0,(p_{pub,eq}-1)\cdot c_{i}^{N}(public,p_{eq})+(1-p_{pub,eq})\cdot c_{i}^{N}(private,p_{eq}))+\\
%%\max(0,p_{pub,eq}\cdot c_{i}^{N}(public,p_{eq})-p_{pub,eq}\cdot c_{i}^{N}(private,p_{eq})) = 0 \nonumber
%%\end{eqnarray}
%%\normalsize
%%Equivalently,
%%\small
%%\begin{eqnarray}
%%(p_{pub,eq}-1)\cdot c_{i}^{N}(public,p_{eq})+(1-p_{pub,eq})\cdot c_{i}^{N}(private,p_{eq}) &\leq& 0 \text{ and } \nonumber \\
%%p_{pub,eq}\cdot c_{i}^{N}(public,p_{eq})-p_{pub,eq}\cdot c_{i}^{N}(private,p_{eq}) &\leq& 0 \nonumber
%%\end{eqnarray}
%%\normalsize
%%Since $0<p_{pub,eq}<1$, we have that

\small
\begin{equation}\label{eq:cost_equat}
c_{i}^{N}(public,p^{NE}) = c_{i}^{N}(private,p^{NE})
\end{equation}
\normalsize

namely, the costs of each pure action belonging to the support of the equilibrium mixed-action strategy
%%\ie they are employed with non-zero probability,
are equal. Hence, from (\ref{eq:costs_sym}) and (\ref{eq:cost_equat}) the symmetric mixed-action equilibrium
$p^{NE}=(p^{NE}_{pub},p^{NE}_{priv})$ solves the equation

\vspace{-8pt}
\small
\begin{eqnarray}\label{eq:sym_alter}
f(p)=-\beta+\sum\limits_{k=0}^{N-1}(\gamma-\min(1,\frac{R}{k+1})\cdot(\gamma-1)) B(k;N-1,p)=0
\end{eqnarray}
\normalsize
%\small
%%$f^{N}(p)=-a+\sum\limits_{\sigma_{pub}=0}^{N-1}(b-(\min(\frac{R}{\sigma_{pub}+1},1)\cdot(b-1)))\cdot %B_{\sigma_{pub}}(N-1,p_{pub})=0$,
%%and \\
%%$B_{\sigma_{pub}}(N-1,p_{pub})={N-1\choose \sigma_{pub}}p_{pub}^{\sigma_{pub}}(1-p_{pub})^{N-1-\sigma_{pub}}$
%%\normalsize
%
A closed-form expression for the equilibrium $p^{NE}_{pub}$ is not straightforward. However, it holds that:
\small
\begin{equation}\label{eq:sym_alter_bounds}
%\lim_{p\rightarrow 0}f(p)&=&-\gamma+1<0 \text{ and }\\
%\lim_{p\rightarrow 1}f(p)&=&\delta-\frac{\delta\sigma_0}{N} > 0
\lim_{p\rightarrow 0}f(p)=-\beta+1<0 \text{ and } \lim_{p\rightarrow 1}f(p)=\delta(1-\frac{\sigma_0}{N}) > 0
\end{equation}
\normalsize
and $f(p)$ is a continuous and strictly increasing function in $p$ since
\small
%\begin{eqnarray}
%f'(p) & = & \sum_{k=0}^{N-1}(\gamma-min(1,\frac{R}{k+1}(\gamma-1))B'(k;N-1,p) \nonumber \\
%      & > & \sum_{k=0}^{N-1}B'(k;N-1,p)= (\sum_{k=0}^{N-1}B(k;N-1,p))' = 0
%\end{eqnarray}
\begin{eqnarray}
f'(p)=\sum_{k=0}^{N-1}(\gamma-\min(1,\frac{R}{k+1})(\gamma-1))B'(k;N-1,p) \nonumber \\
> \sum_{k=0}^{N-1}B'(k;N-1,p)=0 \nonumber
\end{eqnarray}
\normalsize
Hence, $f(p)$ has a single solution. It may be checked with replacement that $f(\sigma_0/N)=0$.
%
%\begin{eqnarray}\label{eq:sym_alter_mono}
%f'%\frac{\vartheta{f}}{p})
%= \sum_{\sigma_{pub}=0}^{N-1}(s(\sigma_{pub}-p(N-1))) \nonumber \\
%                     =\sum_{\sigma_{pub}=0}^{p(N-1)-1}(s(\sigma_{pub}-p(N-1)))+\sum_{\sigma_{pub}=p(N-1)}^{N-1}(s(\sigma_{pub}-p(N-1)))
%
%\end{eqnarray}
%
%\end{proof}

\section{Proof of Theorem 4.1}%\ref{thm:sme_bayes}

%\begin{proof}
Inline with the reasoning in the proof of Theorem \ref{thm:sme}, any symmetric mixed-action equilibrium $p^{NE_{B}}$ must fulfill

\small
\begin{equation}\label{eq:cost_equat_bayes}
c_{i}^{N_B}(public,p^{NE_{B}}) = c_{i}^{N_B}(private,p^{NE_{B}})
\end{equation}
\normalsize
Since

\vspace{-8pt}
\small
\begin{eqnarray}\label{eq:costs_bayesian}
c_{i}^{N_B}(private,p)&=&c_{priv} \nonumber \\
c_{i}^{N_B}(public,p)&=&\sum_{n_{act}=0}^{N-1}c_{i}^{n_{act}+1}(public,p)B(n_{act};N-1,p_{act}) \nonumber
\end{eqnarray}
\normalsize
%Similarly, by (\ref{eq:cost_priv_sym_bayesian}), (\ref{eq:cost_pub_sym_bayesian}) and (\ref{eq:cost_equat_bayes})
a few algebraic manipulations suffice to derive that the symmetric mixed-action equilibrium $p^{NE_{B}}$ solves the equation

\vspace{-15pt}
\small
\begin{eqnarray}\label{eq:sym_alter}
h(p) &=& -\beta+\sum_{n_{act}=0}^{N-1}B(n_{act};N-1,p_{act})\cdot\nonumber \\
&& \sum\limits_{k=0}^{n_{act}}(\gamma-\min(\frac{R}{k+1},1)\cdot(\gamma-1)) B(k;n_{act},p)=0 \nonumber \\
\end{eqnarray}
\normalsize

The function $h(p)$ is continuous and strictly increasing in $p$ for all $p_{act} \in [0,1]$ since
\vspace{-3pt}
\small
\begin{eqnarray}
h'(p) & = & \sum_{n_{act}=0}^{N-1}B(n_{act};N-1,p_{act})\cdot\nonumber \\
&& \sum\limits_{k=0}^{n_{act}}(\gamma-\min(\frac{R}{k+1},1)\cdot(\gamma-1)) B'(k;n_{act},p) \nonumber \\
      & > & \sum_{n_{act}=0}^{N-1}B(n_{act};N-1,p_{act})\sum\limits_{k=0}^{n_{act}}B'(k;n_{act},p) \nonumber \\
      & = & \sum_{n_{act}=0}^{N-1}B(n_{act};N-1,p_{act})(\sum\limits_{k=0}^{n_{act}}B(k;n_{act},p))' = 0 \nonumber
\end{eqnarray}
\normalsize

since the weights of the rightmost Binomial coefficients in the second line are not smaller than one.
\vspace{-3pt}
\small
\begin{equation}
\gamma-\min(\frac{R}{k+1},1)\cdot(\gamma-1)> 1, \forall k \in [0,N-1] \nonumber
\end{equation}
\normalsize

Likewise, for $p_{act} \in [0,1]$, its leftmost value is
\small
\begin{equation}\label{eq:sym_low_bound}
\lim_{p\rightarrow 0}h(p)=-\beta+1<0
\end{equation}
\normalsize
whereas its rightmost value is
\vspace{-5pt}
\small
\begin{eqnarray}\label{eq:sym_upper_bound}
\lim_{p\rightarrow 1}h(p) &=& -\beta+\sum_{n_{act}=0}^{N-1}B(n_{act};N-1,p_{act})\cdot \nonumber \\
&& (\gamma-\min(1,\frac{R}{n_{act}+1})\cdot(\gamma-1)) \nonumber \\
&=& f(p_{act})
\end{eqnarray}
\normalsize

In the proof of Theorem \ref{thm:sme} we showed that the function $f(p)$ is strictly increasing in $p$ and has a single solution $p=\sigma_0/N$. Therefore, as long as $p_{act} \in [0,\sigma_0/N)$, $\lim_{p\rightarrow 1}h(p) < 0$, and $c_{i}^{N_B}(public,p) < c_{i}^{N_B}(private,p)$ $\forall p \in (0,1)$; namely, it is a dominant strategy for all drivers to compete for on-street parking. On the contrary, for $p_{act} \in [\sigma_0/N,1]$, $\lim_{p\rightarrow 1}h(p)$ gets positive values and $h(p)=0$ has a single solution $p=\frac{\sigma_{0}}{Np_{act}}$ (can be checked with replacement).
%
%%Now, for $p_{act}<\frac{\sigma_{0}}{\bar{N}}$, $h^(p)<0  \forall p$
%%The upper bound of
%%
%%For $N>\sigma_{0}$ ,the limits of $h(p)$ as $p$ approaches $0$ and $1$ ar
%%
%%However, we have already shown that for $\bar{N}>\sigma_{0}$, $f^{\bar{N}}(p_{act})$ is a strictly increasing function in $p_{act}$ and that it has a single solution at $p_{act}=\frac{\sigma_{0}}{\bar{N}}$. Hence, for $p_{act}<\frac{\sigma_{0}}{\bar{N}}$, $h^{\bar{N}}(p)<0$, for all $p$, which means that the action $public$ is dominant.
%%However, for $p_{act}\geq\frac{\sigma_{0}}{\bar{N}}$, $\lim_{p\rightarrow 1}h^{\bar{N}}(p)\geq0$. Furthermore, $h^{\bar{N}}(p)$ is a strictly increasing function in $p$ since
%%
%%Therefore, $h^{\bar{N}}(p)$ has a single solution. It may be checked (\textbf{with replacement}) as well that $h(\frac{\sigma_{0}}{\bar{N}p_{act}})=0$.
%
%%For $\bar{N}\leq\sigma_{0}$, the function $h^{\bar{N}}(p)$ is negative for all $p$ and hence again the action $public$ %dominates over $private$.
%\end{proof} 
%\section{Headings in Appendices}
%\balancecolumns
% That's all folks!
\end{document}